\documentclass[journal]{IEEEtran}

\ifx\UseOption\undefined
\def\UseOption{opt2} 
\fi

\ifodd 0
\newcommand{\rev}[1]{{\color{blue}#1}}
\else
\newcommand{\rev}[1]{#1}
\fi

\usepackage{amssymb,amsmath,color}
\usepackage{subfigure}
\usepackage{cite}
\usepackage{verbatim}
\usepackage{algorithm}
\usepackage{algorithmic}
\usepackage{optional}
\usepackage{graphicx,epstopdf}
\usepackage{url}

\DeclareMathOperator*{\argmin}{arg\,min}

\DeclareMathOperator*{\mini}{minimize}

\newtheorem{theorem}{Theorem}
\newtheorem{assumption}{Assumption}

\newtheorem{definition}{Definition}

\newtheorem{lemma}{Lemma}

\begin{document}

\title{DAWN: Delay-Aware Wi-Fi Offloading and Network Selection}

\author{Man Hon Cheung and Jianwei Huang, \IEEEmembership{Senior Member, IEEE}
\thanks{This work is supported by the General Research Funds (Project Number CUHK 412713 and 14202814) established under the University Grant Committee of the Hong Kong Special Administrative Region, China. Part of this paper was presented in \cite{cheung_od13}.}
\thanks{M. H. Cheung and J. Huang are with the Department of Information Engineering, the Chinese University of Hong Kong, Hong Kong, China, \hspace{0.5cm}
e-mail: \{mhcheung, jwhuang\}@ie.cuhk.edu.hk.}
}

\maketitle

\begin{abstract}
  To accommodate the explosive growth in mobile data traffic, both mobile cellular operators and mobile users are increasingly interested in offloading the traffic from cellular networks to Wi-Fi networks.
  However, previously proposed offloading schemes mainly focus on reducing the cellular data usage, without paying too much attention on the quality of service (QoS) requirements of the applications.
  In this paper, we study the Wi-Fi offloading problem with delay-tolerant applications under usage-based pricing. We aim to achieve a good tradeoff between the user's payment and its QoS characterized by the file transfer deadline.
  We first propose a general \textbf{D}elay-\textbf{A}ware \textbf{W}i-Fi Offloading and \textbf{N}etwork Selection (DAWN) algorithm for a general single-user decision scenario.
  We then analytically establish the sufficient conditions, under which the optimal policy exhibits a threshold structure in terms of both the time and file size.
  As a result, we propose a monotone DAWN algorithm that approximately solves the general offloading problem, and has a much lower computational complexity comparing to the optimal algorithm. 
  Simulation results show that both the general and monotone DAWN schemes achieve a high probability of completing file transfer under a stringent deadline, and require the lowest payment under a non-stringent deadline as compared with three heuristic schemes.
\end{abstract}


\section{Introduction} \label{sec:intro}

  \PARstart{M}{obile} cellular networks nowadays are often heavily loaded due to the huge amount of mobile data traffic generated, for example, through mobile web browsing and mobile video applications.
  According to Cisco's forecast, mobile data traffic will increase by 11-fold between 2013 and 2018 globally \cite{cisco_cv14}.
  On the other hand, the mobile cellular network capacity is growing at a much slower pace, so that it is likely that the mobile traffic demand will exceed the network capacity in the short to medium term \cite{lucent_wr12}.
  As a result, there is an urgent need from the mobile operators (MOs) worldwide to increase the network capacity in a cost-effective and timely manner. 
  An efficient way to ease the cellular congestion is to use complementary technologies, such as Wi-Fi \cite{disruptiveanalysis_cw11}, 
  to offload the traffic originally targeted towards the cellular network.
  Juniper Research estimated that only 40\% of the global mobile data traffic will reach the cellular network in 2017, as most of the traffic are likely to be offloaded using Wi-Fi \cite{juniper_md13}. 
  
  There are two main approaches for the initiation of Wi-Fi offloading, namely user-initiated and operator-initiated offloading. In the \emph{user-initiated} offloading, the mobile user (MU) is responsible for selecting the network technologies that it intends to use. In the \emph{operator-initiated} offloading, however, the operator profile stored in the mobile device prompts the connection manager to initiate the offloading procedure. 
  The MOs would prefer the operator-initiated offloading, as it gives them a better control on users' network selections. However, since the operator-initiated offloading involves complicated network control between the MOs and the MUs, further standardization and development are still under way.
	Currently, the user-initiated offloading is the more popular choice due to its simplicity in implementation, and it will be the focus of this paper.

  New functionalities in some recently proposed IEEE and 3GPP architectures can provide MUs with useful network information for the user-initiated offloading.
  In Hotspot 2.0, which is based on the IEEE 802.11u standard \cite{ieee80211u_std}, the network discovery and selection functionality advertises the network information related to the access network type, roaming consortium, and venue information through management frames.
  The access network discovery and selection function (ANDSF) server \cite{3gpp_ts23402, lucent_wr12}, proposed in 3GPP Release 11, can assist an MU to choose a suitable Wi-Fi network by providing it with a list of preferred access networks and the access network discovery information. 
  Moreover, it is envisaged that network information, such as real-time load and radio conditions, can be broadcast to the MUs through the system information block (SIB) messages currently used in the LTE system \cite[pp.\,46]{4gamericas_io13}.
  With these new architectures for cellular and Wi-Fi integration, MUs can make intelligent network selection and offloading decision based on real-time network load and price information.

  In addition, the standardization effort from the industry has been accompanied by a series of efforts on the characterization of \emph{Wi-Fi offloading} performance from the academia. Recently, measurement studies \cite{balasubramanian_am10, lee_md10} demonstrated that Wi-Fi offloading can significant reduce the cellular network congestion.       
  In fact, the potential benefit of data offloading is even more significant \cite{balasubramanian_am10, lee_md10} for \emph{delay-tolerant} applications, such as e-mail, movie download, and software update, which can tolerate delays ranging from several minutes to several hours without significant negative impact on users' satisfactions.
  For example, the survey in \cite{sen_wt13} reported that more than half of the respondents are willing to wait for 10 minutes to stream YouTube videos and 3-5 hours to download a file when a monetary incentive is given.

  In this paper, we study the user-initiated Wi-Fi offloading problem for delay-tolerant applications, where a user aims to minimize its total data usage payment under \emph{usage-based pricing}, while taking into account the deadline of its application.
  Previous works on user-initiated Wi-Fi offloading policy, which includes \cite{rayment_ac12, lee_md10, balasubramanian_am10, ristanovic_ee11}, mainly focus on reducing the cellular data usage without paying too much attention to the quality of service (QoS) of the user's application.
  As an example, consider the on-the-spot offloading (OTSO) scheme that most smartphones are using by default \cite{rayment_ac12}. The OTSO scheme adopts a simple offloading policy that an MU offloads its data traffic to a Wi-Fi network whenever possible.    
  However, our simulation study suggests that it is not always desirable to offload to Wi-Fi whenever possible, especially when the Wi-Fi network is highly loaded and the deadline is tight. 
  However, in general it is challenging to achieve a good balance between the total payment and the QoS when taking various factors such as network conditions and delay deadlines into consideration.   

  First, we consider a general user offloading scenario, and formulate the delay-aware Wi-Fi offloading problem as a finite-horizon sequential decision problem. We propose a general \textbf{D}elay-\textbf{A}ware \textbf{W}i-Fi Offloading and \textbf{N}etwork Selection (DAWN) algorithm, which achieves a good tradeoff between the total payment and the QoS.
  However, in general, a sequential decision problem is computationally intractable unless the optimal policy has a threshold structure \cite{ngo_oo09}.
  To this end, based on the concepts of \emph{superadditivity} and \emph{subadditivity} \cite[pp.\,103]{puterman_md05}, we derive sufficient conditions under which the optimal policy exhibits threshold structures in terms of both time and the remaining file size to transfer.
  It motivates us to design the monotone DAWN algorithm with a much lower computational complexity that approximately solves the general offloading problem. 
  To the best of our knowledge, this is the first paper that studies offloading algorithm design analytically, which tradeoffs a user's payment and QoS. 
  The insights obtained, even under the single-user setting in the user-initiated offloading, are crucial for us to understand the more complicated multi-user offloading problems in commercial networks.

  In summary, the main contributions of our work are as follows:

\begin{itemize}

\item \emph{Optimal user-initiated offloading algorithm}: We consider the Wi-Fi offloading problem for delay-tolerant applications, and propose a general DAWN algorithm that achieves a good tradeoff between total data usage payment and the user's QoS.

\item \emph{Low-complexity approximation offloading algorithm}: We derive sufficient conditions under which the optimal policy has a threshold structure, and propose a monotone approximation DAWN algorithm with a much lower computational complexity. 

\item \emph{Optimal offloading decisions}:  Simulation results show that the general and monotone DAWN algorithms achieve a high probability of file transfer completion and require a low payment as compared with three heuristic schemes. We also show that Wi-Fi offloading may not be desirable under a tight deadline constraint and a congested Wi-Fi network. 

\end{itemize}

  The rest of the paper is organized as follows.
	We first review the literature in mobile data offloading in Section \ref{sec:literature}.
  We describe our system model in Section \ref{sec:model}, and formulate the delay-aware Wi-Fi offloading problem in Section \ref{sec:prob}. 
  We propose the general DAWN algorithm for the general case in Section \ref{sec:dp}, and the monotone DAWN algorithm for the special case with threshold optimal policy in Section \ref{sec:threshold}.
  Simulation results are given in Section \ref{sec:pe}, and the paper is concluded in Section \ref{sec:concl}.

\section{Literature Review} \label{sec:literature}
  
  The existing mobile data offloading literature focuses on either economics or technology  issues. 
  Related to network economics,
  Zhuo \emph{et al.} in \cite{zhuo_wc11} considered a 3G cellular network, where the MO uses discount coupons to incentivize MUs to use delayed data offloading. The problem was formulated as a reverse auction with one buyer and multiple sellers, where the MO is the buyer, and the MUs are the sellers.  
  Joe-Wong \emph{et al.} in \cite{joewong_os13} studied the user adoption of supplementary technology (e.g., Wi-Fi or femtocell) for cellular traffic offloading. The utility function of each user is related to its valuation of the technology, the congestion level, and the flat pricing of the service provider. 
  The studies in \cite{gao_eo13, iosifidis_ai13} considered an offloading market, where the MOs pay the third-party deployed APs for data offloading. 
  Gao \emph{et al.} in \cite{gao_eo13} characterized the subgame perfect equilibrium in a data offloading game, where the base stations (BSs) propose the market prices, and the APs determine the volume of data traffic that they are willing to offload.
  Iosifidis \emph{et al.} in \cite{iosifidis_ai13}  proposed an iterative and incentive compatible double auction that maximizes the social welfare. 
  Lee \emph{et al.} in \cite{lee_eo13} studied the economic aspects of Wi-Fi offloading in a monopolistic market with multiple MUs and one MO. Each MU is characterized by its willingness to pay, traffic demand, delay profile, and Wi-Fi contact probability. 

  Related to the mobile data offloading technology,
  Dimatteo \emph{et al.} in \cite{dimatteo_ct11} evaluated the costs and benefits of Wi-Fi offloading in metropolitan area with real mobility traces. They characterized the number of Wi-Fi access points (APs) required for the support of a given QoS requirement.	
	Bennis \emph{et al.} in \cite{bennis_wc13} studied the subband selection, power allocation, and scheduling problem of a small cell base station, which can transmit with both the cellular and Wi-Fi interfaces. The base stations can self-organize and adjust their transmission strategies using reinforcement learning. 
	There are a number of recent research results on the study of \emph{delayed Wi-Fi offloading} policy.	
  Balasubramanian \emph{et al.} in \cite{balasubramanian_am10} conducted a measurement study on Wi-Fi availability for moving vehicles. They proposed the Wiffler system for data offloading based on the prediction for future Wi-Fi availability using past mobility history.	
	Lee \emph{et al.} in \cite{lee_md10} performed another measurement study on Wi-Fi offloading with pedestrians. They conducted trace-driven simulations to study the impact of various parameters on the offloading efficiency.
  Ristanovic \emph{et al.} in \cite{ristanovic_ee11} considered energy-efficient offloading for delay-tolerant applications. They showed that the proposed offloading algorithms can offload a significant amount of traffic from the cellular network and extend the battery lifetime.  
  Im \emph{et al.} in \cite{im_ae13} considered the cost-throughput-delay tradeoff in user-initiated Wi-Fi offloading. Given the predicted future usage and the availability of Wi-Fi, the proposed system decides on the application that should offload its traffic to Wi-Fi at a given time, while taking into account the cellular budget constraint of the MU.

  In fact, similar to this paper with a detailed user's decision model, the works related to data offloading algorithm design in \cite{lee_md10, balasubramanian_am10, ristanovic_ee11, im_ae13} focus on the single-user offloading problem.
  On the other hand, the works related to data offloading economics in \cite{zhuo_wc11, joewong_os13, gao_eo13, iosifidis_ai13, lee_eo13} considered simplified models on users' decisions, and they mainly focus on the multi-user offloading problem.

\section{System Model} \label{sec:model}

\begin{figure}[t]
 \centering
   \includegraphics[width=7.5cm, trim = 1cm 3cm 2cm 2cm, clip = true]{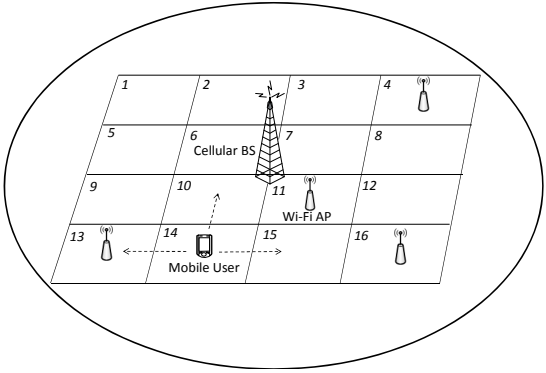} 
 \caption{An example of the network setting, where the MU is moving within a set of $\mathcal{L} = \{1, \ldots, 16\}$ locations. 
 The MU is always under the coverage of a cellular BS, but Wi-Fi is only available at four locations, where $\mathcal{L}^{(1)} = \{4,11,13,16\}$. The rest of the locations do not have Wi-Fi, i.e., $\mathcal{L}^{(0)} = \mathcal{L} \backslash \mathcal{L}^{(1)}$. We assume that the MU is sending a file of size $K$ bits that should be completed by deadline $T$. Given the mobility pattern of the MU, it aims to decide whether it should remain idle ($a=0$), use the cellular network ($a=1$), or use the Wi-Fi network ($a=2$) if it is available in each time slot to reduce its payment under usage-based pricing, while taking into account the deadline of the application.}
\label{fig:network}
\end{figure}

  As shown in Fig. \ref{fig:network}, we consider an MU\footnote{In this paper, since we focus on the user-initiated offloading, it is reasonable to consider the setting that a single MU makes an independent decision without coordinating with the other MUs. We believe that it is an important step towards a better understanding of the multi-user offloading problem in the operator-initiated offloading, where the MO needs to decide on the offloading decisions of multiple MUs.} moving within the coverage of the cellular network, such that the cellular connection is always available to the MU. Occasionally, the MU may be able to access Wi-Fi APs at some locations (e.g., in a coffee shop or in a shopping mall). In other words, the Wi-Fi connection is \emph{location-dependent} and may not be available to the MU at all time.
  The MU is running a file transfer application, which requires transferring of $K$ bits within $T$ time slots. In other words, the file transfer application is \emph{delay-tolerant} with a deadline $T$ \cite{sen_wt13}. For example, an MU on the road wants to send an e-mail with a large attachment of $20$ Mbytes through his smartphone in the next $10$ minutes.
  The MU moves in a set $\mathcal{L} = \{1, \ldots, L\}$ of possible locations, following a Markovian mobility model that can be derived based on the past mobility pattern of the MU. Such a model is widely used in the literature \cite{nicholson_bf08, niyato_em10, gambs_np12}. 

  We consider the \emph{usage-based pricing} used by MOs (such as the one used by Verizon Wireless \cite{verizon}), where the usage price of the cellular network is often higher than that of the Wi-Fi network. It should be noted that the pricing scheme is general, and it includes free Wi-Fi as a special case.
  When making the offloading decisions, the MU needs to take into account the payments regarding different network types and its QoS requirement in terms of file transfer completion.
  First, the MU has the incentive to offload as much data traffic to the Wi-Fi network as possible, so as to reduce its payment. This means that the MU prefers to defer the transmission until a Wi-Fi hotspot is available. 
  On the other hand, the MU should also consider whether it can complete the file transfer by the deadline. For example, if the remaining time before the deadline is short, then the deferred transmission may violate the deadline if the MU does not have enough opportunities to transmit through Wi-Fi in the near future. In this case, instead, the MU should start the file transfer using the ubiquitous cellular connection as soon as possible to reduce the latency.  
  To sum up, an efficient delay-aware Wi-Fi offloading scheme needs to achieve a good \emph{tradeoff} between the total data usage payment and the MU's QoS requirement.

	As the Wi-Fi offloading problem involves decision making in multiple time slots before the deadline, we formulate it as a finite-horizon sequential decision problem in the following section.
  We aim to find the MU's optimal transmission policy, which minimizes the MU's data usage payment, while taking into account the deadline of the file transfer application. 
  By defining the total \emph{cost} as the total payment and a \emph{penalty} for not finishing the file transfer by the deadline, we can derive the optimal transmission policy through dynamic programming (DP).
  We further propose an approximation algorithm based on a non-standard DP theory.

\section{Problem Formulation} \label{sec:prob}

  In this section, we formulate the delay-aware Wi-Fi offloading problem of a \emph{single} MU as a \emph{finite-horizon sequential} decision problem \cite{puterman_md05}. 
  Without loss of generality, we normalize the length of a time slot to be one.
  The MU needs to choose an action (to be explained later) at each \emph{decision epoch}
\begin{equation}
  t \in \mathcal{T} = \{1, \ldots, T\}.
\end{equation}
%

  The system \emph{state} is defined as $\boldsymbol{s} = (k, l)$.     
  The state element $k \in \mathcal{K} \subseteq [0,K]$ represents the \emph{remaining} size (in bits) of a file\footnote{\rev{For the case where the MU is transferring multiple files, we can include additional state elements and decisions, and solve the problem by dynamic programming.}} to be transferred.  
  The state element $l \in \mathcal{L} = \{1, \ldots, L\}$ is the location index, where $L$ is the total number of possible locations that the MU may reach within the $T$ time slots.
  As shown in Fig. \ref{fig:network}, let $\mathcal{L}^{(0)} \subseteq \mathcal{L}$ and $\mathcal{L}^{(1)} \subseteq \mathcal{L}$ be the sets of locations where Wi-Fi is not and is available, respectively, such that $\mathcal{L}^{(0)} = \mathcal{L} \backslash \mathcal{L}^{(1)}$.

   The \emph{action} $a$ specifies the transmission decision of the MU at each decision epoch.
   Specifically, we have $a \in \mathcal{A} = \{0,1,2\}$, where $a = 0$ means that the MU chooses to remain idle, $a = 1$ means that the MU transmits through cellular, and $a = 2$ represents that the MU transmits through Wi-Fi.
   Notice that actions $a = 0$ and $a = 1$ are always available to the MU at all locations. Action $a = 2$, however, is only available at a location $l \in \mathcal{L}^{(1)}$. Thus, the available choice of action $a$ depends on the state element $l$, so $a \in \mathcal{A}^{(l)} \subseteq \mathcal{A}$, where $\mathcal{A}^{(l)}$ is the set of available transmission actions at location $l$: 
\begin{equation} \label{equ:setal}
  \mathcal{A}^{(l)} = 
\begin{cases}
  \{0,1,2\}, & \mbox{if } l \in \mathcal{L}^{(1)}, \\  
  \{0,1\},   & \mbox{if } l \in \mathcal{L}^{(0)}.
\end{cases}  
\end{equation}

  We adopt the commonly used \emph{usage-based pricing}, where the payment of an MU is directly proportional to its data usage. Let $p(l,a)$ be the price per unit of usage for choosing action $a  \in \mathcal{A}^{(l)}$ at location $l$, where $p(l, 0) = 0, \, \forall \, l \in \mathcal{L}$ for the idle action.
	It should be noted that we consider a general location and network dependent pricing, which includes the commonly used location independent pricing as a special case. 
  Let $\mu(l,a)$ be the estimated throughput of the user at location $l$ with action $a \in \mathcal{A}^{(l)}$, where $\mu(l, 0) = 0, \, \forall \, l \in \mathcal{L}$ when the MU remains idle (i.e., when $a = 0$). 
	\rev{We would like to mention that $\mu(l,a)$ can take into account the congestion effect when multiple MUs are simultaneously using the same network\footnote{\rev{For the detailed study of strategic network selection interactions among multiple MUs, we refer readers to our work in \cite{cheung_ca14}.}}.
	We assume that the MU can obtain such real-time price and data rate\footnote{\rev{By allowing Hotspot 2.0 and ANDSF to complement with each other \cite{4gamericas_io13, lucent_wr12}, an MU can query for the speed and load in different types of networks \cite{ruckus_hi13} before transmitting data in these networks.}} information for accessing networks at different time and locations through the system information block (SIB) announced by the MO, as discussed in Section \ref{sec:intro} \cite{4gamericas_io13}.} 
  The \rev{\emph{payment}} of the MU at state $\boldsymbol{s}$ with action $a \in \mathcal{A}^{(l)}$ at time slot $t \in \mathcal{T}$ is
\begin{equation} \label{equ:payment}
  c_t(\boldsymbol{s}, a) = c_t(k, l, a) = \min \{k, \mu(l,a)\} p(l,a),
\end{equation}
which is equal to the data usage payment in the time slot. 

  After the deadline has passed, we define the \emph{penalty} for not being able to finish the file transfer at state $\boldsymbol{s}$ as
\begin{equation} \label{equ:penalty}
 \hat{c}_{T+1}(\boldsymbol{s}) = \hat{c}_{T+1}(k,l) = h(k),
\end{equation}
where $h(k) \geq 0$ is a non-decreasing function of $k$ with $h(0) = 0$. The subscript $T+1$ means that we compute the penalty at the beginning of the $T+1$ time slot (immediately after the deadline). In fact, the MU chooses $h(k)$ according to the QoS requirement of its application. 

  The \emph{state transition probability} $p \bigl( \boldsymbol{s}'  \,|\, \boldsymbol{s}, a \bigr) = p \bigl( (k',l') \,|\, (k,l), a \bigr)$ is the probability that the system will go into state $\boldsymbol{s}' = (k',l')$ in the next time slot if action $a$ is taken at state $\boldsymbol{s} = (k,l)$.
  Since the movement of the MU from location $l$ to location $l'$ is independent of the file size $k$ and transmission action $a$, we have  
\begin{equation} \label{equ:stp}
  p \bigl( \boldsymbol{s}'  \,|\, \boldsymbol{s}, a \bigr) = p \bigl( (k',l') \,|\, (k,l), a \bigr) 
= p(l' \,|\, l) \; p \bigl( k' \,|\, (k,l), a \bigr), 
\end{equation}
where 
\begin{equation} \label{equ:stp_k}
  p \bigl( k' \,|\, (k,l), a \bigr) = 
\begin{cases} 
  1, & \mbox{if } k' = [k - \mu(l,a)]^+ \mbox{ and } a \in \mathcal{A}^{(l)},\\
  0, & \mbox{otherwise,} 
\end{cases}  
\end{equation}
and $[x]^+ = \max \{0,x\}$. 
  $p(l' \,|\, l)$ is the probability\footnote{\rev{Prototype systems, such as BreadCrumbs \cite{nicholson_bf08}, can compute the movement probability by tracking the movement of the device's owner.}} that the MU will move from location $l$ to location $l'$, and it is estimated based on the past mobility pattern of the MU \cite{nicholson_bf08, niyato_em10, gambs_np12}.

	Let $\delta_t: \mathcal{K} \times \mathcal{L} \rightarrow \mathcal{A}$ be a function that specifies the transmission decision of the MU at state $\boldsymbol{s} = (k,l)$ and time slot $t$. We define a \emph{policy} $\boldsymbol{\pi} = (\delta_t(k,l), \, \forall \, k \in \mathcal{K}, l \in \mathcal{L}, t \in \mathcal{T})$ as the set of decision rules for states and time slots.
  We denote $\boldsymbol{s}_t^{\boldsymbol{\pi}} = (k_t^{\boldsymbol{\pi}}, l_t^{\boldsymbol{\pi}})$ as the state at time slot $t$ if policy $\boldsymbol{\pi}$ is used, and we let $\Pi$ be the feasible set of $\boldsymbol{\pi}$. 
  The MU aims to find an optimal policy $\boldsymbol{\pi}^*$ that minimizes the sum\footnote{\rev{For simplicity, we assume that the total payment and penalty have equal weights. If we put a larger weight on the penalty than the total payment, then the probability of completing file transfer would increase, at the expense of an increase in the payment.}} of the expected total \rev{payment} from $t = 1$ to $t = T$ and the penalty at $t = T + 1$ as follows:
\begin{equation} \label{equ:problem}
\begin{array}{rll}
     \displaystyle \mini_{\boldsymbol{\pi} \in \Pi} & \displaystyle E_{\boldsymbol{s}_1}^{\boldsymbol{\pi}} \Biggl[ \sum_{t=1}^T  c_t \bigl( \boldsymbol{s}_t^{\boldsymbol{\pi}}, \delta_t(\boldsymbol{s}_t^{\boldsymbol{\pi}}) \bigr) + \hat{c}_{T+1}(\boldsymbol{s}_{T+1}^{\boldsymbol{\pi}}) \Biggr]. & \\
\end{array}
\end{equation}
$E_{\boldsymbol{s}_1}^{\boldsymbol{\pi}}$ denotes the expectation with respect to the probability distribution of the MU mobility model and policy $\boldsymbol{\pi}$ with an initial state $\boldsymbol{s}_1 = (K, l_1)$, where $l_1$ is the location of the MU at $t = 1$.

\section{General DAWN Algorithm} \label{sec:dp} 

   In this section, we solve problem \eqref{equ:problem} \emph{optimally} using \emph{finite-horizon} DP for the general penalty function, network usage prices, and cellular/Wi-Fi data rates. We propose a general DAWN algorithm that computes the optimal policy.

  Let $v_t(\boldsymbol{s})$ be the minimal expected total cost\footnote{\rev{Here, we use the term \emph{cost} to represent both the payment in \eqref{equ:payment} at time $t \in \mathcal{T}$ and the penalty in \eqref{equ:penalty} at time $T+1$.}} of the MU from time slot $t$ to $T+1$, given that the system is in state $\boldsymbol{s}$ immediately before the decision at time slot $t$. 
  The \emph{optimality equation} \cite[pp.\,83]{puterman_md05} relating the minimal expected total cost at different states for $t \in \mathcal{T}$ is given by
\begin{equation} \label{equ:opteq}
	v_t(\boldsymbol{s}) = v_t(k,l) = \displaystyle \min_{a \in \mathcal{A}^{(l)}} \{ \psi_t(k, l, a) \},
\end{equation}
where for $k \in \mathcal{K}$, $l \in \mathcal{L}$, and $a \in \mathcal{A}^{(l)}$, we have
\opt{opt1}{ 
\begin{eqnarray} 
  \psi_t(k,l,a) = c_t(k,l,a) + \sum_{l' \in \mathcal{L}} \, \sum_{k' \in \mathcal{K}}  p\bigl( (k',l') \,|\,(k,l),a \bigr) \, v_{t+1}(k',l')  & \label{equ:opteq2} \\
  = \min \{k, \mu(l,a)\} p(l,a) +  \sum_{l' \in \mathcal{L}} p(l' \,|\, l) \, v_{t+1} \bigl( [k - \mu(l,a)]^+, l' \bigr). \hspace{-1.45cm} & \label{equ:opteq2_newcost}
\end{eqnarray}
}
\opt{opt2}{ 
\begin{eqnarray} 
  \psi_t(k,l,a)  \hspace{6.3cm} & \nonumber \\
  = c_t(k,l,a) + \sum_{l' \in \mathcal{L}} \, \sum_{k' \in \mathcal{K}}  p\bigl( (k',l') \,|\,(k,l),a \bigr) \, v_{t+1}(k',l') \hspace{0.2cm} & \label{equ:opteq2} \\
  = \! \min \! \{k, \mu(l,a)\} p(l,a) \! + \!\!\! \sum_{l' \in \mathcal{L}} p(l' | l) v_{t+1} \! \bigl( [k \! - \! \mu(l,a)]^+\!\!, l' \!\bigr). & \label{equ:opteq2_newcost}
\end{eqnarray}
}
  The first and second terms on the right hand side of \eqref{equ:opteq2} are the \emph{immediate cost} and the \emph{expected future cost} in the remaining time slots for choosing action $a$, respectively. 
The derivation of \eqref{equ:opteq2_newcost} from \eqref{equ:opteq2} follows directly from \eqref{equ:payment}, \eqref{equ:stp}, and \eqref{equ:stp_k}.
  For $t = T+1$, we set the boundary condition as
\begin{equation} \label{equ:boundary}
	v_{T+1}(\boldsymbol{s}) = \hat{c}_{T+1}(k,l) = h(k), \quad \forall \, k \in \mathcal{K}, l \in \mathcal{L}.
\end{equation}
%


\begin{algorithm} [t] \small
\caption{\emph{General Delay-Aware Wi-Fi Offloading and Network Selection (DAWN) Algorithm.}}
\begin{algorithmic} [1] \label{algo:offload}

\STATE \underline{Planning Phase}:

\STATE Set $v_{T+1}(k,l), \, \forall \, k \in \mathcal{K}, \forall \, l \in \mathcal{L}$ using \eqref{equ:boundary}

\STATE Set $t := T$

\STATE \textbf{while} $t \geq 1$

\STATE $ \ \ \ $ \textbf{for} $l \in \mathcal{L}$

\STATE $ \ \ \ \ \ \ $ Set $k := 0$

\STATE $ \ \ \ \ \ \ $ \textbf{while} $k \leq K$

\STATE $ \ \ \ \ \ \ \ \ \ $ Calculate $\psi_t(k,l,a), \, \forall \, a \in \mathcal{A}^{(l)}$ using \eqref{equ:opteq2_newcost}

\STATE $ \ \ \ \ \ \ \ \ \ $ Set $\displaystyle \delta_t^{*}(k,l) := \argmin_{a \in \mathcal{A}^{(l)}} \{ \psi_t(k,l,a) \}$

\STATE $ \ \ \ \ \ \ \ \ \ $ Set $v_t(k,l) := \psi_t \bigl( k,l, \delta_t^{*}(k,l) \bigr)$

\STATE $ \ \ \ \ \ \ \ \ \ $ Set $k := k + \sigma$

\STATE $ \ \ \ \ \ \ $ \textbf{end while}

\STATE $ \ \ \ $ \textbf{end for}

\STATE $ \ \ \ $ Set $t := t - 1$

\STATE \textbf{end while}

\STATE Output the optimal policy $\boldsymbol{\pi}^*$ for the transmission and Wi-Fi offloading phase

\STATE \underline{Transmission and Wi-Fi Offloading Phase}:

\STATE Set $t := 1$ and $k := K$

\STATE \textbf{while} $t \leq T$ \textbf{and} $k > 0$

\STATE $ \ \ \ $ Determine the location index $l$ from GPS

\STATE $ \ \ \ $ Set action $a := \delta_t^{*}(k,l)$ based on the optimal policy $\boldsymbol{\pi}^*$

\STATE $ \ \ \ $ \textbf{If} $a > 0$

\STATE $ \ \ \ \ \ \ $ Send $\mu(l,1)$ bits to the cellular network if $a = 1$ \\ \quad\quad\quad or offload $\mu(l,2)$ bits to the Wi-Fi network if $a = 2$


\STATE $ \ \ \ \ \ \ $ Set $k := [k - \mu(l,a)]^+$


\STATE $ \ \ \ $ \textbf{end if}

\STATE $ \ \ \ $ Set $t := t + 1$

\STATE \textbf{end while}

\end{algorithmic}
\end{algorithm}    
  


  With the optimality equation, we are ready to propose the general DAWN algorithm in Algorithm \ref{algo:offload}. The algorithm consists of two phases, namely the planning phase and the transmission and Wi-Fi offloading phase. 
  Let $\sigma > 0$ be the granularity of the discrete state element $k$ in the algorithm (such as $1$ Mbits).
  First, in the planning phase, based on the optimality equation in \eqref{equ:opteq} and the boundary condition in \eqref{equ:boundary}, we obtain the \emph{optimal policy} $\boldsymbol{\pi}^*$ that solves problem \eqref{equ:problem} using \emph{backward induction} \cite[pp.\,92]{puterman_md05}.
    Specifically, we first set $v_{T+1}(k,l)$ based on the boundary condition (line 2) of Algorithm \ref{algo:offload}. Then, we obtain the values of $\delta_t^{*}(k,l)$ and $v_t(k,l)$ by updating them recursively backward from time slot $t=T$ to time slot $t=1$ (lines 3 to 16).
    Algorithm \ref{algo:offload} has a computational complexity of  $\mathcal{O}(K L T / \sigma)$ \cite{ngo_oo09}. 


\vspace{0.2cm}
\begin{theorem} \label{thm:dp}
  The policy $\boldsymbol{\pi}^* = (\delta_t^*(k,l), \, \forall \, k \in \mathcal{K}, l \in \mathcal{L}, t \in \mathcal{T})$, where
\begin{equation} \label{equ:deltat}
	\delta_t^{*}(k,l) = \argmin_{a \in \mathcal{A}^{(l)}} \{ \psi_t(k,l,a) \}, 
\end{equation}
is the optimal solution of problem \eqref{equ:problem}. 
\end{theorem}

\vspace{0.2cm}
\begin{proof}
  Using the principle of optimality \cite[pp.\,18]{bertsekas_dp05}, we can show that $\boldsymbol{\pi}^*$ is the optimal solution of problem \eqref{equ:problem}.
\end{proof}
\vspace{0.2cm}

  Notice that the optimal policy $\boldsymbol{\pi}^*$ is a \emph{contingency plan} that contains information about the optimal transmission decision at \emph{all} the possible states $(k,l)$ in any time slots $t \in \mathcal{T}$, and the system computes it \emph{offline} before the file transfer begins in the second phase. 
  In the second phase, the MU first determines the location index $l$ in each time slot based on the location information obtained by global positioning system (GPS) (line 20). Then, the MU carries out the transmission decisions based on the optimal policy $\boldsymbol{\pi}^*$ through checking a table (lines 21 to 25), and updates the state element $k$ accordingly (line 24).
  As the complexity of Algorithm \ref{algo:offload} is high in general, it motivates us to design an approximation algorithm with a lower computational complexity in the next section.

\section{Threshold Policy and Monotone DAWN Algorithm} \label{sec:threshold}

  In this section, we establish sufficient conditions under which the optimal policy has a \emph{threshold} structure in the remaining file size $k$ and time $t$. We then propose a monotone DAWN algorithm accordingly, which approximately solves problem \eqref{equ:problem} in the general case with a lower computational complexity.
  Thus, the results cannot be obtained by a direct application of the standard DP theory.

  Specifically, we make the following assumptions for deriving the optimal policy in this section:
  
\begin{assumption} \label{ass:threshold}
(a) The penalty function $h(k)$ is convex and non-decreasing in $k$; 
(b) Wi-Fi is free to the MU (i.e., $p(l,2) = 0, \, \forall \, l \in \mathcal{L}^{(1)}$);
(c) The cellular price is location-independent (i.e., $p(l,1) = p(l',1), \, \forall \, l, l' \in \mathcal{L}, l \neq l'$); 
(d) The cellular and Wi-Fi data rates are location-independent (but these two rates are different in general). That is, $\mu_1 = \mu(l,1), \, \forall \, l \in \mathcal{L}$ and $\mu_2 = \mu(l,2), \, \forall \, l \in \mathcal{L}^{(1)}$; and 
(e) We approximate $\min \{k, \mu(l,1)\}$ in \eqref{equ:payment} by $\mu(l,1)$ for action $a = 1$.
\end{assumption}

  Notice that (a) a convex penalty function can be used to model the increasing marginal penalty for every additional unit of file segment not yet transferred. It is similar to the idea that a concave utility function can be used to model the diminishing marginal utility.
  (b) Free Wi-Fi can often be found in places such as homes, offices, or coffee shops.
  (c) Location-independent cellular price is widely used in practice.
  (d) is a good approximation when the cellular and Wi-Fi data rates across different locations have a small variance.
	(e) is a technical approximation for simplifying the structure of the optimal policy.

	
	 With Assumption \ref{ass:threshold}, the cost at state $\boldsymbol{s}$ with action $a$ at time slot $t$ is modified from \eqref{equ:payment} as
\begin{equation} \label{equ:cost}
  c_t(\boldsymbol{s}, a) = c_t(k, l, a) = I(a = 1)q = 
\begin{cases} 
  q, & \mbox{if } a = 1, \\
  0, & \mbox{otherwise,} 
\end{cases}  
\end{equation}
where $a \in \mathcal{A}^{(l)}$, $I(\cdot)$ is the indicator function, and $q = \mu(l,1) \, p(l,1)$. As a result, $\psi_t(k,l,a)$ in \eqref{equ:opteq2} can be rewritten as
\begin{equation} \label{equ:opteq2_2}
	\psi_t(k,l,a) = I(a = 1)q +  \sum_{l' \in \mathcal{L}} p(l' \,|\, l) \, v_{t+1} \bigl( [k - \mu(l,a)]^+, l' \bigr).
\end{equation}
%
%
%

\subsection{Properties of the Optimal Policy}

  First, we discuss some analytical results related to the properties of the optimal policy under Assumption \ref{ass:threshold}.  


\begin{lemma} \label{lem:vnondeck}
  (a) $v_t(k,l)$ is a non-decreasing function in $k$, $\forall \, l \in \mathcal{L}, t \in \mathcal{T}$.
  (b) $v_t(k,l)$ is a non-decreasing function in $t$, $\forall \, k \in \mathcal{K}, l \in \mathcal{L}$.
\end{lemma}

  The proof of Lemma \ref{lem:vnondeck} is given in Appendix \ref{app:vnondeck}. 
  Intuitively, given a fixed location $l \in \mathcal{L}$, the expected cost is higher when $k$ is larger (i.e., the remaining file size to transfer is larger) or when $t$ is larger (i.e., it is closer to the deadline).  
    
  Next, we characterize the optimal transmission policy at a location $l \in \mathcal{L}^{(1)}$ with Wi-Fi.
  Since Wi-Fi is free for use, Lemma \ref{lem:wifi}(a) states that action $a = 2$ (i.e., using Wi-Fi) is always preferred to action $a = 0$ (i.e., remaining idle). 
  Lemma \ref{lem:wifi}(b) states that if the Wi-Fi data rate is higher than the cellular data rate, then the MU should always use Wi-Fi.

\vspace{0.2cm}
\begin{lemma} \label{lem:wifi}
  For any location $l \in \mathcal{L}^{(1)}$ (where Wi-Fi is available), we have:\\
  (a) $\psi_t(k,l,0) \geq \psi_t(k,l,2), \, \forall \, k \in \mathcal{K}, t \in \mathcal{T}$. \\
  (b) If $\mu(l,1) \leq \mu(l,2)$, then $\delta_t^*(k,l) = 2, \, \forall \, k \in \mathcal{K}, t \in \mathcal{T}$.
\end{lemma}
\vspace{0.2cm}

  The proof of Lemma \ref{lem:wifi} is given in Appendix \ref{app:wifi}.
  Notice that at $l \in \mathcal{L}^{(1)}$, although $\mathcal{A}^{(l)} = \{0,1,2\}$ from \eqref{equ:setal}, Lemma \ref{lem:wifi}(a) implies that we do not need to consider action $a = 0$ in \eqref{equ:opteq}. Specifically, let
\begin{equation} \label{equ:setal_2}
  \tilde{\mathcal{A}}^{(l)} = 
\begin{cases}
  \{1,2\}, & \mbox{if } l \in \mathcal{L}^{(1)}, \\  
  \{0,1\}, & \mbox{if } l \in \mathcal{L}^{(0)}.
\end{cases}  
\end{equation}
  We can simplify the optimality equation in \eqref{equ:opteq} as
\begin{equation} \label{equ:opteq_2}
	v_t(k,l) = \min_{a \in \mathcal{A}^{(l)}} \{ \psi_t(k,l,a) \} = \min_{a \in \tilde{\mathcal{A}}^{(l)}} \{ \psi_t(k,l,a) \}.
\end{equation}
%

\subsection{Threshold Structure of the Optimal Policy}

  To show the threshold policy in dimension $k$, we need to leverage on the concepts of \emph{superadditivity} and \emph{subadditivity} \cite[pp.\,103]{puterman_md05}. 
  Specifically, with the assumptions we made on the penalty function and data rates, 
  we show in Appendix \ref{app:additive} that $\psi_t(k,l,a)$ is superadditive or subadditive on $\mathcal{K} \times \tilde{\mathcal{A}}^{(l)}$ under different conditions. 
  Then, with $\delta_t^{*}(k,l)$ defined in (\ref{equ:deltat}), we can establish the threshold structure of the optimal policy in dimension $k$ \cite[pp.\,104, 115]{puterman_md05}.  

\vspace{0.2cm}  
\begin{definition} \label{def:additive}
  Given $l \in \mathcal{L}$, the function $\psi_t(k,l,a)$ is \emph{superadditive} on $\mathcal{K} \times \tilde{\mathcal{A}}^{(l)}$ if for $\forall \, \hat{k}, \check{k} \in \mathcal{K}$ and $\forall \, \hat{a}, \check{a} \in \mathcal{A}$, where $\hat{k} \geq \check{k}$ and $\hat{a} \geq \check{a}$, we have
\begin{equation} \label{equ:superadditive}
   \psi_t( \hat{k} ,l, \hat{a} )  + \psi_t( \check{k} ,l, \check{a} ) \geq \psi_t( \hat{k} ,l, \check{a} ) + \psi_t( \check{k} ,l, \hat{a} ).
\end{equation}  
The function $\psi_t(k,l,a)$ is \emph{subadditive} on $\mathcal{K} \times \tilde{\mathcal{A}}^{(l)}$ if the reverse inequality always holds.
\end{definition}   
\vspace{0.2cm}

  To prove the threshold policy in dimension $t$, we show in Appendix \ref{app:vtkl_increment} that the incremental changes of $v_t(k,l)$ with respect to $k$ is non-decreasing in time $t$.
  Overall, we state the threshold policy in both dimensions $k$ and $t$ as follows.

\vspace{0.3cm}
\begin{theorem} \label{thm:threshold}
  Under Assumption \ref{ass:threshold}, the optimal policy $\boldsymbol{\pi}^* = (\delta_t^{*}(k,l), \, \forall \, k \in \mathcal{K}, l \in \mathcal{L}, t \in \mathcal{T})$ has a \emph{threshold} structure in both $k$ and $t$ as follows: 
  
  For location $l \in \mathcal{L}^{(0)}$ without Wi-Fi, we have 
\begin{equation} \label{equ:threshold0}
  \delta_t^{*}(k,l) = 
\begin{cases}
  1 \; \text{ (cellular)},  & \mbox{if } k \geq k^*(l,t), \\
  0 \; \text{  (idle)},  & \mbox{otherwise,} 
\end{cases}
\, \forall \, t \in \mathcal{T},	\text{ and}
\end{equation}
%
%
%
\begin{equation} \label{equ:threshold_time0}
  \delta_t^{*}(k,l) =  
\begin{cases}
  1 \; \text{ (cellular)},  & \mbox{if } t \geq t^*(k,l), \\
  0 \; \text{ (idle)},  & \mbox{otherwise,} 
\end{cases}
\, \forall \, k \in \mathcal{K}, \hspace{0.9cm}
\end{equation}
%
%
where $k^*(l,t)$ and $t^*(k,l)$ are location and time dependent thresholds in dimensions $k$ and $t$, respectively.

  For location $l \in \mathcal{L}^{(1)}$ with Wi-Fi, 
  if the data rate of Wi-Fi is lower than that of cellular (i.e., $\mu_2 \leq \mu_1$), we have 
%
\begin{equation} \label{equ:threshold2}
  \delta_t^{*}(k,l) = 
\begin{cases}
  1 \; \text{ (cellular)},  & \mbox{if } k \geq k^*(l,t), \\
  2 \; \text{ (Wi-Fi)},  & \mbox{otherwise,} 
\end{cases}
\, \forall \, t \in \mathcal{T},	\text{ and}
\end{equation}
%
%
%
\begin{equation} \label{equ:threshold_time2}
  \delta_t^{*}(k,l) = 
\begin{cases}
  1 \; \text{ (cellular)},  & \mbox{if } t \geq t^*(k,l), \\
  2 \; \text{ (Wi-Fi)},  & \mbox{otherwise,} 
\end{cases}
\, \forall \, k \in \mathcal{K}. \hspace{0.9cm}
\end{equation}
%
%
Otherwise (hence $\mu_1 < \mu_2$), we have 
\begin{equation} \label{equ:threshold2_2}
	\delta_t^{*}(k,l) = 2 \; \text{ (Wi-Fi)}, \, \forall \, k \in \mathcal{K}, t \in \mathcal{T}.
\end{equation}
\end{theorem} 

  Theorem \ref{thm:threshold} states that when $k$ is above a threshold (i.e., there are many bits waiting to be transmitted) or when $t$ is above a threshold (i.e., the deadline is close), the MU should use the cellular network immediately to avoid the penalty (if Wi-Fi is not available or Wi-Fi is not fast enough).
  The proof of threshold policy in dimension $k$ stated in \eqref{equ:threshold0} and \eqref{equ:threshold2} is given in Appendix \ref{app:threshold}.
  The proof of threshold policy in dimension $t$ stated in \eqref{equ:threshold_time0} and \eqref{equ:threshold_time2} is given in Appendix \ref{app:threshold_time}.  
  The result in \eqref{equ:threshold2_2} is due to Lemma \ref{lem:wifi}(b).

  Furthermore, we use the threshold structure in Theorem \ref{thm:threshold} to establish Theorem \ref{thm:threshold_t}, which help to speed up the search of the thresholds. 
  
\begin{theorem} \label{thm:threshold_t}
  (a) $k^*(l,t-1) \geq k^*(l,t), \; \forall \, l \in \mathcal{L}, t \in \mathcal{T}$. \\
  (b) $t^*(k, l) \geq t^*(k + \sigma, l), \; \forall \, l \in \mathcal{L}, k \in \mathcal{K}$.
\end{theorem}  

  Basically, Theorem \ref{thm:threshold_t} states that the threshold in dimension $k$ is non-increasing in $t$, while the threshold in dimension $t$ is non-increasing in $k$.
  The proof of Theorem \ref{thm:threshold_t} is given in Appendix \ref{app:threshold_t}.

\subsection{Monotone DAWN Algorithm}

  With the threshold structure in both dimensions $k$ and $t$ from Theorems \ref{thm:threshold} and \ref{thm:threshold_t}, we propose Algorithm \ref{algo:offload_threshold} with a much lower computational complexity than Algorithm \ref{algo:offload}. 
  In Algorithm \ref{algo:offload_threshold}, it should be noted that we choose to characterize the optimal policy $\boldsymbol{\pi}^*$ using the thresholds $(k^*(l,t),  \, \forall \, l \in \mathcal{L}, t \in \mathcal{T})$ in the file size dimension in \eqref{equ:threshold0} and \eqref{equ:threshold2}.  
  In the planning phase of Algorithm \ref{algo:offload_threshold}, we use the procedure \texttt{THRESHOLD} to obtain the set of thresholds $(k^*(l,t),  \, \forall \, l \in \mathcal{L}, t \in \mathcal{T})$ (line 10) in dimension $k$. 
  Since we execute the algorithm backward from $t = T$ to $t = 1$, after we have found the threshold $k^*(l,t)$ at time $t$, we can reduce the search space of $k^*(l,t-1)$ at time $t-1$ by Theorem \ref{thm:threshold_t}(a).


\begin{algorithm} [t] \small
\caption{\emph{Monotone DAWN Algorithm.}} 
\begin{algorithmic} [1] \label{algo:offload_threshold}

\STATE \underline{Planning Phase} (for $\mu_2 \leq \mu_1$):

\STATE Set $v_{T+1}(k,l), \, \forall \, k \in \mathcal{K}, l \in \mathcal{L}$ using \eqref{equ:boundary}

\STATE Set $t := T$

\STATE \textbf{while} $t \geq 1$

\STATE $ \ \ \ $ \textbf{for} $l \in \mathcal{L}$

\STATE $ \ \ \ \ \ \ $ Call \texttt{THRESHOLD} procedure

\STATE $ \ \ \ $ \textbf{end for}

\STATE $ \ \ \ $ Set $t := t - 1$

\STATE \textbf{end while}

\STATE Output the thresholds $(k^*(l,t),  \, \forall \, l \in \mathcal{L}, t \in \mathcal{T})$ for the transmission and Wi-Fi offloading phase

\STATE \underline{Transmission and Wi-Fi Offloading Phase}:

\STATE Set $t := 1$ and $k := K$

\STATE \textbf{while} $t \leq T$ \textbf{and} $k > 0$

\STATE $ \ \ \ $ Determine the location index $l$ from GPS 


\STATE $ \ \ \ $ \textbf{If} $l \in \mathcal{L}^{(0)}$						

\STATE $ \ \ \ \ \ \ $ \textbf{If} $k \geq k^*(l,t)$, Set $a := 1$, \textbf{else}, Set $a := 0$, \textbf{end if}

%

\STATE $ \ \ \ $ \textbf{else if} $l \in \mathcal{L}^{(1)}$

\STATE $ \ \ \ \ \ \ $ \textbf{If} $\mu_2 \leq \mu_1$

\STATE $ \ \ \ \ \ \ \ \ \ $ \textbf{If} $k \geq k^*(l,t)$, Set $a := 1$, \textbf{else}, Set $a := 2$, \textbf{end if}

\STATE $ \ \ \ \ \ \ $ \textbf{else}

\STATE $ \ \ \ \ \ \ \ \ \ $ Set $a := 2$

\STATE $ \ \ \ \ \ \ $ \textbf{end if}

\STATE $ \ \ \ $ \textbf{end if}

\STATE $ \ \ \ $ \textbf{If} $a > 0$

\STATE $ \ \ \ \ \ \ $ Send $\mu(l,1)$ bits to the cellular network if $a = 1$ \\ \quad\quad\quad or offload $\mu(l,2)$ bits to the Wi-Fi network if $a = 2$


\STATE $ \ \ \ \ \ \ $ Set $k := [k - \mu(l,a)]^+$


\STATE $ \ \ \ $ \textbf{end if}

\STATE $ \ \ \ $ Set $t := t + 1$

\STATE \textbf{end while}
\end{algorithmic}

\bigskip

\textbf{procedure} \texttt{THRESHOLD}
\begin{algorithmic}[1]

\STATE \textbf{If} $l \in \mathcal{L}^{(0)}$, Set $j := 0$, \textbf{else}, Set $j := 2$, \textbf{end if}

\STATE Set $\mathcal{A}^{\text{th}} := \{j\}$, $k := 0$, and $flag := 0$

\STATE \textbf{while} $k \leq K$

\STATE $ \ \ \ $ \textbf{if} $k \geq k^*(l, t+1)$ and $flag = 0$

\STATE $ \ \ \ \ \ \ $ Set $\mathcal{A}^{\text{th}} := \{j,1\}$ and $flag := 1$

\STATE $ \ \ \ $ \textbf{end if}

\STATE $ \ \ \ $ Calculate $\psi_t(k,l,a), \, \forall \, a \in \mathcal{A}^{\text{th}}$ using \eqref{equ:opteq2_2}

\STATE $ \ \ \ $ Set $\displaystyle \delta_t^{*}(k,l) := \argmin_{a \in \mathcal{A}^{\text{th}}} \{ \psi_t(k,l,a) \}$

\STATE $ \ \ \ $ Set $v_t(k,l) := \psi_t \bigl( k,l, \delta_t^{*}(k,l) \bigr)$

\STATE $ \ \ \ $ \textbf{if} $\delta_{t}^{*}(k,l) = 1$ and $flag = 1$

\STATE $ \ \ \ \ \ \ $ Set $\mathcal{A}^{\text{th}} := \{1\}$, $k^*(l,t) := k$, and $flag := 2$

\STATE $ \ \ \ $ \textbf{end if}

\STATE $ \ \ \ $ Set $k := k + \sigma$

\STATE \textbf{end while}
  
\end{algorithmic}

\end{algorithm}

	By knowing the threshold structure in both dimensions $k$ and $t$ in Theorems \ref{thm:threshold} and \ref{thm:threshold_t}, we can speed up the computation of the optimal policy.
  Let $\mathcal{A}^{\text{th}} \subseteq \tilde{\mathcal{A}}^{(l)}$ be the set of feasible actions that we should consider (procedure line 8) for the optimal policy.  
  Instead of considering the two possible actions in $\tilde{\mathcal{A}}^{(l)} = \{1,2\}$ for $l \in \mathcal{L}^{(1)}$ and $\tilde{\mathcal{A}}^{(l)} = \{0,1\}$ for $l \in \mathcal{L}^{(0)}$ in \eqref{equ:setal_2} for \eqref{equ:opteq_2}, we can reduce the amount of computation by only considering one possible action in $\mathcal{A}^{\text{th}}$ under two conditions: 
  (i) When $k < k^*(l, t+1)$, we know from Theorem \ref{thm:threshold_t}(a) that $k < k^*(l, t)$, so we only need to consider $\mathcal{A}^{\text{th}} = \{j\}$ with one element (procedure line 2).
  (ii) When we have reached the threshold that $k > k^*(l, t)$, we know from Theorem \ref{thm:threshold} that we only need to consider $\mathcal{A}^{\text{th}} = \{1\}$ with one element (procedure line 11). 
  In both cases (i) and (ii), $\mathcal{A}^{\text{th}}$ becomes a singleton, and the minimization in line 8 of the procedure is readily known. As a result, the computational complexity is reduced from $\mathcal{O}(K L T / \sigma)$ in Algorithm \ref{algo:offload} to approximately $\mathcal{O}(L \max\{K/\sigma,T \})$ in Algorithm \ref{algo:offload_threshold} \cite{ngo_oo09}. 
  In the second phase, we determine action $a$ based on the threshold optimal policy in dimension $k$ stated in Theorem \ref{thm:threshold}. Specifically, the decisions in lines 16, 19, and 21 are due to \eqref{equ:threshold0}, \eqref{equ:threshold2}, and \eqref{equ:threshold2_2}, respectively.

\section{Performance Evaluations} \label{sec:pe}

  In this section, we evaluate the performance of the general and monotone DAWN schemes by comparing them with three benchmark schemes (the no offloading, on-the-spot offloading \cite{lee_md10}, and Wiffler \cite{balasubramanian_am10} schemes) in terms of the total cost, probability of completing file transfer, and the total payment.
  We also illustrate the threshold policy stated in Theorem \ref{thm:threshold}.
  
  For each set of system parameter choices, we run the simulations 1000 times with randomized Wi-Fi locations, data rates in the cellular and Wi-Fi networks, and the user mobility trajectories, and show the average value.
  The MU is moving within $L = 16$ possible locations in a four by four grid (similar to that in Fig. \ref{fig:network}). 
  To generate the trajectory of the MU, we consider the state transition probabilities $p(l' \,|\, l)$, where we assume that probability that the MU stays at a location between two consecutive time slots is $p(l \,|\, l) = 0.6, \, \forall \, l \in \mathcal{L}$. Moreover, it is equally likely for the MU to move to any one of the neighbouring locations. As an example in Fig. \ref{fig:network}, at location $7$, the probability that the MU will move to one of the locations $3$, $6$, $8$, or $11$ is equal to $(1-0.6)/4 = 0.1$.  
  For another example, at location $1$, the probability that the MU will move to one of the neighbouring locations $2$ and $5$ is equal to $(1-0.6)/2 = 0.2$.  

  
	Unless specified otherwise, we assume that the cellular data rate $\mu(l,1), \forall \, l \in \mathcal{L}$ and the Wi-Fi data rate $\mu(l,2), \forall \, l \in \mathcal{L}^{(1)}$ are truncated (on the range of $[0,\infty)$) normally distributed random variables with means $\mu_c$ and $\mu_w$, respectively, and standard deviations equal to $5$ Mbps.  
  We assume that the cellular usage price $p(l,1), \, \forall \, l \in \mathcal{L}$ is US \$$6$/Gbyte, while the Wi-Fi is free such that $p(l,2) = 0, \, \forall \, l \in \mathcal{L}^{(1)}$.
  The probability that a Wi-Fi connection is available at a particular location is $0.5$.
  The length of a time slot $\Delta t$ equals to ten seconds.
  We consider that the MU is transferring a file (e.g., a movie), where the deadline of the file transfer is $D$ minutes (so $T = 60 D / \Delta t)$.
  We set the file size granularity $\sigma = 10$ Mbits.
  For the delay violation penalty, we use the convex function
\begin{equation} \label{equ:convexpenalty}
	h(k) = b \, k^2, \quad \forall \, k \in \mathcal{K},
\end{equation}
where $b \geq 0$ is a constant. As an example, we adopt $b = 1$ in our simulations.

  Next we explain the five schemes in our simulations.
  Under the general DAWN scheme, we run the planning phase in Algorithm \ref{algo:offload} with the complete and accurate data rate information $\mu(l,1), \, \forall \, l \in \mathcal{L}$ and $\mu(l,2), \, \forall \, l \in \mathcal{L}^{(1)}$.
  For the monotone DAWN algorithm, however, we assume that the MU only knows the mean data rates in the networks. As a result, we run the planning phase in Algorithm \ref{algo:offload_threshold} with incomplete data rate information by letting $\mu(l,1) = \mu_c, \, \forall \, l \in \mathcal{L}$ and $\mu(l,2) = \mu_w, \, \forall \, l \in \mathcal{L}^{(1)}$.
  Under the \emph{no offloading} scheme, the MU uses the cellular network at all times.
  For the \emph{on-the-spot offloading} (OTSO) scheme, the data traffic is offloaded to the Wi-Fi network whenever Wi-Fi is available. The MU will use the cellular connection immediately when Wi-Fi is not available. 
  The \emph{Wiffler} scheme is a prediction-based offloading scheme proposed in \cite{balasubramanian_am10}. Let $\zeta$ be the estimated amount of data that can be transferred using Wi-Fi by the deadline. The Wiffler system uses a history-based predictor, which estimates $\zeta$ based on the inter-meeting time and throughput of the last $m$ Wi-Fi AP encounters.
  If Wi-Fi is available in the current location, then Wi-Fi will be used immediately. If Wi-Fi is not available, the MU needs to check whether the condition $\zeta \geq \theta k$ is satisfied. Here, $k$ is the remaining size of the file to be transferred, and $\theta > 0$ is the conservative coefficient that tradeoffs the amount of data offloaded with the completion time of the file transfer. If this condition is satisfied, meaning that the estimated data transfer using Wi-Fi is large enough, then the MU will stay idle and wait for the Wi-Fi connection. Otherwise, the MU will use the cellular connection. We set $\theta = 1$ and $m = 4$ as suggested in \cite{balasubramanian_am10}.

\subsection{Comparisons Among Different Schemes}

\begin{figure}[t]
\centering
\includegraphics[width=7cm]{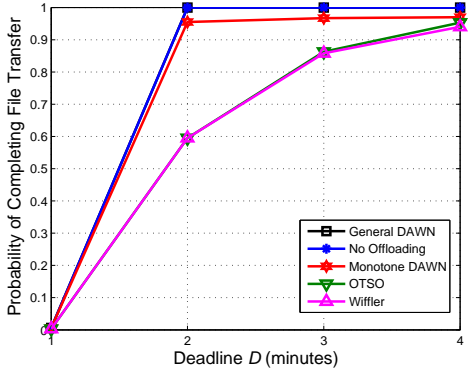}
\caption{The probability of completing file transfer versus deadline $D$ for $p(l \,|\, l) = 0.6, \, \forall \, l \in \mathcal{L}$, $K = 750$ Mbytes, $\mu_c = 90$ Mbps, and $\mu_w = 20$ Mbps.} 
\label{fig:probdeadline}
\end{figure}

\begin{figure}[t]
\centering
\includegraphics[width=7cm]{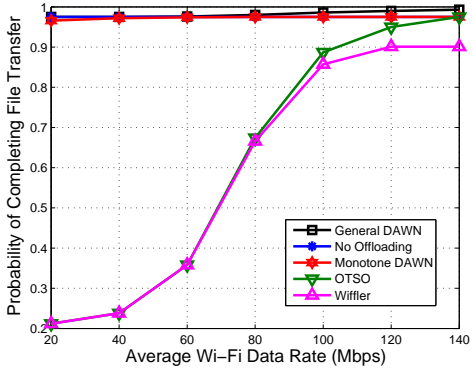} 
\caption{The probability of completing file transfer versus average Wi-Fi data rate $\mu_w$ for $p(l \,|\, l) = 0.6, \, \forall \, l \in \mathcal{L}$, $K = 625$ Mbytes, deadline $D = 1$ min, and $\mu_c = 90$ Mbps.}
\label{fig:probwifirate}
\end{figure}

  In this subsection, we compare the performance of the five schemes (two proposed in this paper and three benchmark schemes) under stringent and non-stringent deadline requirements.
  First, we consider a larger file size $K = 750$ Mbytes, which is challenging to complete the transmission when the deadline is short. 
  Here, we first focus on the special case, where the cellular data rates are much higher than the Wi-Fi data rates. Specifically, we consider that the mean cellular and Wi-Fi data rates are $\mu_c = 90$ Mbps and $\mu_w = 20$ Mbps, respectively, which are reasonable parameters under a 4G LTE-A cellular system \cite{wiki_4g} and a congested Wi-Fi network \cite{ieee80211std}.
  Our later simulation results will show how the performance of the algorithms changes with different Wi-Fi data rates.   
  In Fig. \ref{fig:probdeadline}, we plot the probability of completing file transfer against deadline $D$. As $D$ increases, it is more likely to finish the file transfer before the deadline, so the probability of completing file transfer of the five schemes increases.
  Moreover, we observe that the general DAWN and no offloading schemes achieve the highest probability of completing file transfer, and the monotone DAWN scheme achieves a slightly lower probability. 
  On the other hand, we observe that the OTSO and Wiffler schemes are not able to complete the file transfer around 40\% of time when $D = 2$ mins. 
   The reason is that these two schemes always offload the traffic to the Wi-Fi networks whenever Wi-Fi is available. However, they ignore the QoS requirement of the application in terms of the stringent deadline. When the cellular data rate is higher than the Wi-Fi data rate, it may be preferable to use the cellular network to increase the chance of file transfer completion despite of the higher payment.  
   
  Since the result in Fig. \ref{fig:probdeadline} depends on the relative data rates in the cellular and Wi-Fi networks, we evaluate the probability of completing file transfer against the average Wi-Fi data rate $\mu_w$ under fixed average cellular data rate $\mu_c = 90$ Mbps for $K = 625$ Mbytes and deadline $D = 1$ min in Fig. \ref{fig:probwifirate}.
  \rev{As we can see, when $\mu_w$ increases, the probability of completing file transfer of the OTSO scheme approaches to that of the general DAWN scheme. It is because when the Wi-Fi data rate is much higher than the cellular data rate, OTSO becomes the optimal offloading decision, as Wi-Fi networks are free and have a higher data rate than the cellular network.}

\begin{figure}[t]
\centering
\includegraphics[width=7cm]{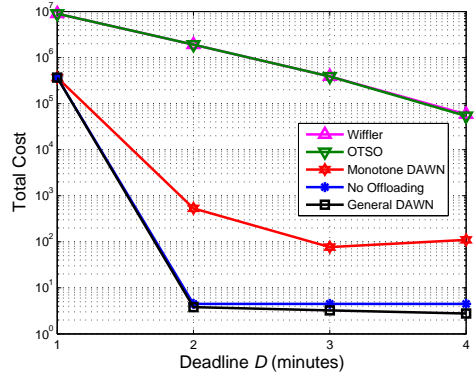}
\caption{Total cost versus deadline $D$ for $p(l \,|\, l) = 0.6, \, \forall \, l \in \mathcal{L}$, $K = 750$ Mbytes, $\mu_c = 90$ Mbps, and $\mu_w = 20$ Mbps.} 
\label{fig:costdeadline}
\end{figure}

\begin{figure}[t]
\centering
\includegraphics[width=7cm]{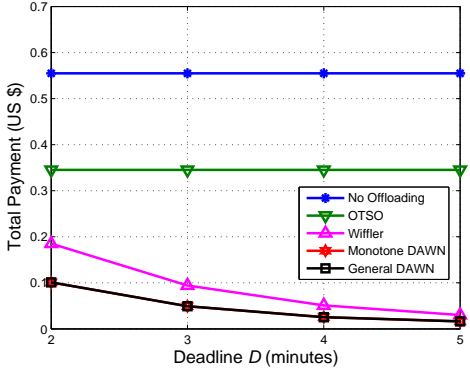}
\caption{The total data usage payment of the user versus deadline $D$ for $p(l \,|\, l) = 0.6, \, \forall \, l \in \mathcal{L}$, $K = 92.5$ Mbytes, $\mu_c = 90$ Mbps, and $\mu_w = 20$ Mbps.} 
\label{fig:paymentdeadline}
\end{figure}

  In Fig. \ref{fig:costdeadline}, we plot the total cost (i.e., the objective function in problem \eqref{equ:problem}) against the deadline $D$ for $K = 750$ Mbytes. Since the general DAWN scheme computes and obtains the optimal policy, it achieves the minimal total cost as stated in Theorem \ref{thm:dp}.  
  Moreover, we observe that the total cost decreases with $D$ for most of the schemes. The reason is that as $D$ increases, the MU has more time to wait for the availability of free Wi-Fi, and thus reduces the total payment. Moreover, for a larger $D$, the chance of completing the file transfer is higher, and the penalty is thus smaller. 
  For the monotone DAWN scheme, however, we observe a slight increase in the total cost at $D = 3$ mins, which is probably due to the incomplete data rate information described above.
  As shown in Fig. \ref{fig:costdeadline}, the general DAWN has a lower total cost than the no offloading scheme, which implies that the general DAWN requires a lower total payment to achieve the same highest probability of completing file transfer as the no offloading scheme does illustrated in Fig. \ref{fig:probdeadline}.

\begin{figure}[t]
 \centering
   \includegraphics[width=7cm]{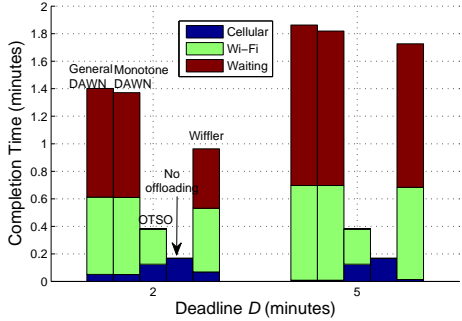} 
 \caption{The average completion time versus deadline $D$ for $K = 92.5$ Mbytes, $\mu_c = 90$ Mbps, and $\mu_w = 20$ Mbps.} 
\label{fig:completiontime}
\end{figure}

  \rev{Next, we consider the case with a non-stringent deadline requirement due to a smaller file size $K = 92.5$ Mbytes, where all the schemes have a very high probability of completing the file transfer in this setup.
  In Fig. \ref{fig:paymentdeadline}, we plot the total payment against deadline $D$ under the five schemes.}
  For the no offloading scheme, since it always uses the more expensive cellular network, its payment is the highest and is independent of $D$.
  For the OTSO scheme, it has a lower payment than the no offloading scheme, because it uses the free Wi-Fi networks whenever they are available. However, the OTSO scheme is not aware of the deadline, so it often incurs a significant penalty for violating the deadline.
  In contrast, the general DAWN, monotone DAWN, and Wiffler schemes are deadline-aware, where they evaluate the chance of file transfer completion by the deadline. When $D$ increases, these three schemes use the Wi-Fi network more often to complete the file transfer, so the total payment decreases.
  In Fig. \ref{fig:paymentdeadline}, we observe that the monotone DAWN achieves the same lowest payment as the general DAWN.

\begin{figure}[t]
\centering
\includegraphics[width=7cm]{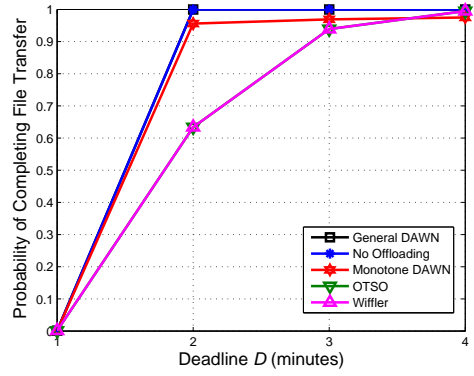}
\caption{The probability of completing file transfer versus deadline $D$ for $p(l \,|\, l) = 0.1, \, \forall \, l \in \mathcal{L}$, $K = 750$ Mbytes, $\mu_c = 90$ Mbps, and $\mu_w = 20$ Mbps.} 
\label{fig:probdeadline_rev}
\end{figure}

\begin{figure}[t]
\centering
\includegraphics[width=7cm]{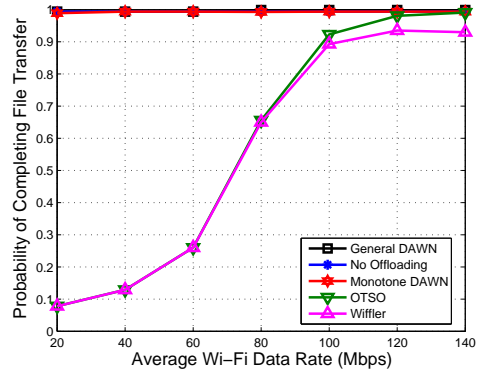} 
\caption{The probability of completing file transfer versus average Wi-Fi data rate $\mu_w$ for $p(l \,|\, l) = 0.1, \, \forall \, l \in \mathcal{L}$, $K = 625$ Mbytes, deadline $D = 1$ min, and $\mu_c = 90$ Mbps.}
\label{fig:probwifirate_rev}
\end{figure}

  \rev{We study in more details on how time is spent before completing the file transfer under a non-stringent deadline requirement as in the setup in Fig. \ref{fig:paymentdeadline}.
	In Fig. \ref{fig:completiontime}, we plot the average completion time of the five schemes for $K = 92.5$ Mbytes under deadlines $D = 2$ mins and $D = 5$ mins. Notice that the completion time includes three parts: cellular transmission time (blue), Wi-Fi transmission time (green), and waiting time for Wi-Fi networks (brown). 
	We can see that the no offloading and OTSO schemes have the shortest lengths of completion time, because of the zero waiting time. 
	On the other hand, the Wiffler, monotone DAWN, and general DAWN schemes experience longer lengths of completion time due to the more significant waiting time. When the deadline is longer ($D = 5$ mins), these three schemes can tolerate a longer delay to wait for the availability of free Wi-Fi networks, and reduce their cellular transmissions, and thus their payments as shown in in Fig. \ref{fig:paymentdeadline}.} 

  \rev{In Figs. \ref{fig:probdeadline_rev}-\ref{fig:paymentdeadline_rev}, we run the simulation experiments in Figs. \ref{fig:probdeadline}-\ref{fig:paymentdeadline} again with different user movement probabilities, where the probability of staying at a location $p(l \,|\, l) = 0.1, \, \forall \, l \in \mathcal{L}$. We can see that the general trends of the curves and insights remain the same, although the magnitude of the performance metrics differ. It suggests that the values of the movement probabilities do not have a significant impact on the insights of our simulation results.}

\begin{figure}[t]
\centering
\includegraphics[width=7cm]{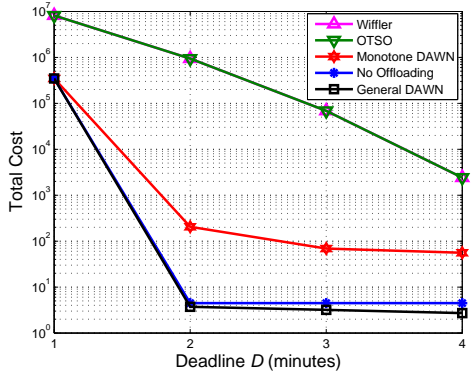}
\caption{Total cost versus deadline $D$ for $p(l \,|\, l) = 0.1, \, \forall \, l \in \mathcal{L}$,  $K = 750$ Mbytes, $\mu_c = 90$ Mbps, and $\mu_w = 20$ Mbps.}
\label{fig:costdeadline_rev}
\end{figure}

\begin{figure}[t]
\centering
\includegraphics[width=7cm]{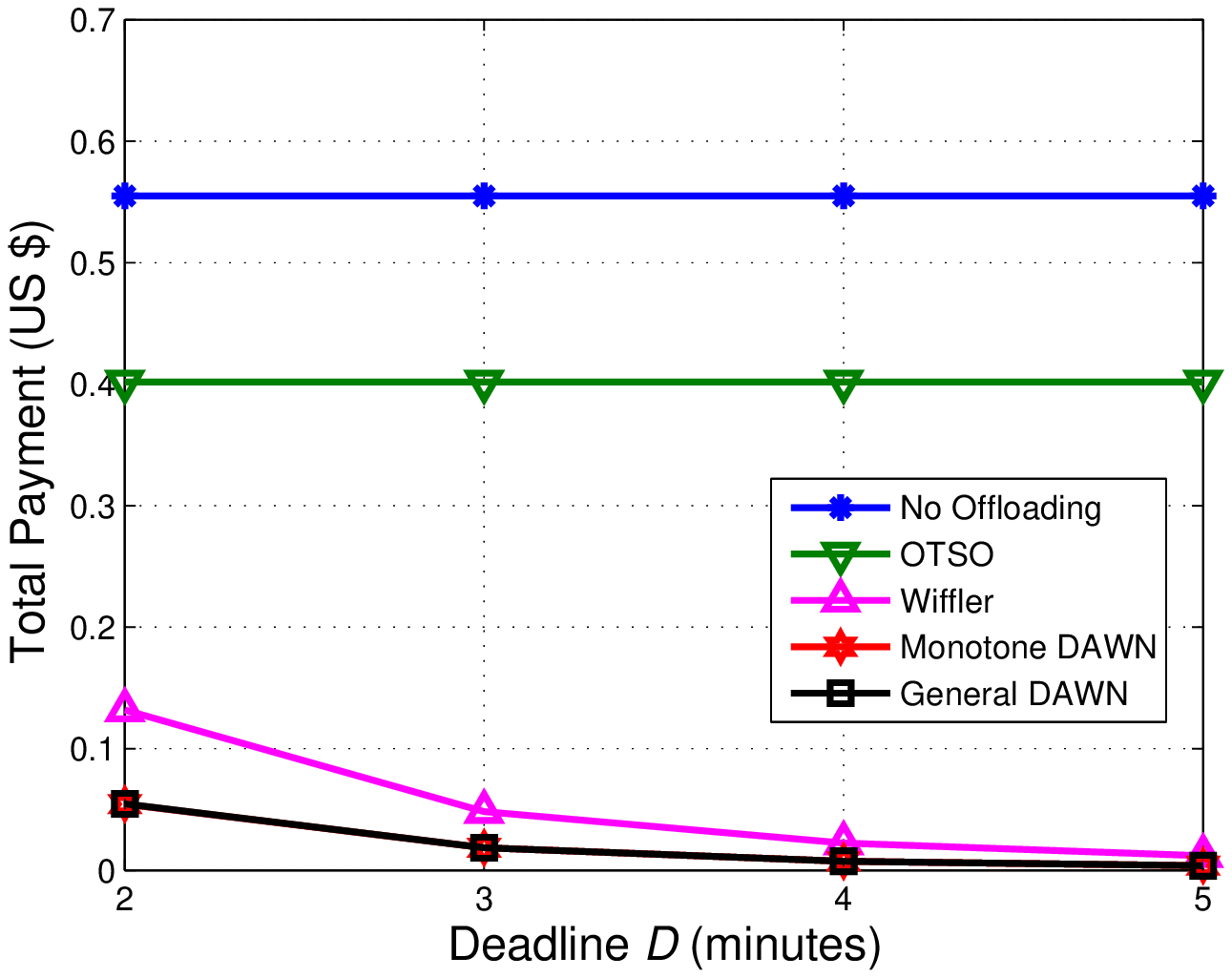}
\caption{The total data usage payment of the user versus deadline $D$ for $p(l \,|\, l) = 0.1, \, \forall \, l \in \mathcal{L}$, $K = 92.5$ Mbytes, $\mu_c = 90$ Mbps, and $\mu_w = 20$ Mbps.}
\label{fig:paymentdeadline_rev}
\end{figure}

\begin{figure}[t] 
\centering
\subfigure[$l \in \mathcal{L}^{(0)}$ for $\mu_1 = 2$ Mbps.]{
   \includegraphics[width=7cm]{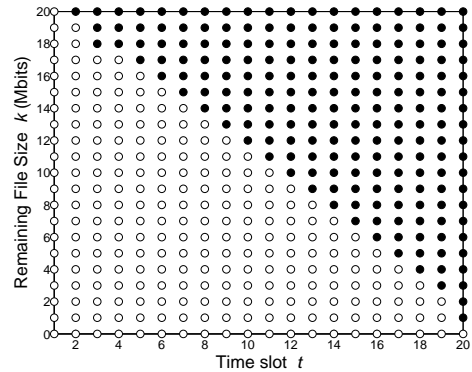}
   \label{fig:threshold_fig1}
 }
\subfigure[$l \in \mathcal{L}^{(1)}$ for $\mu_1 = 2$ Mbps and $\mu_2 = 1$ Mbps.]{
   \includegraphics[width=7cm]{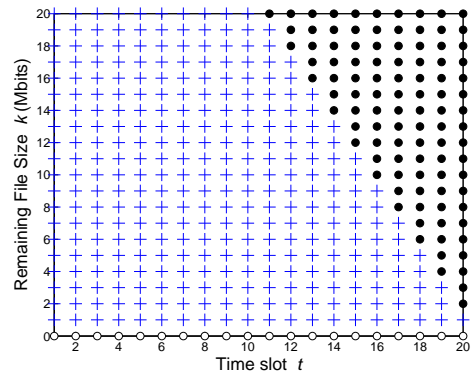}
   \label{fig:threshold_fig2}
 } 
 \caption{An example of the optimal policy at location $l \in \mathcal{L}$ for the case with convex penalty and location-independent data rates, where $K = 20$ Mbits, $\sigma = 1$ Mbits, $T = 20$, and $b = 10$. The white dots ($\circ$), black dots ($\bullet$), and blue crosses (\textcolor{blue}{+}) represent the transmission decisions of $a = 0$ (idle), $1$ (use cellular), and $2$ (use Wi-Fi), respectively. We can observe the threshold optimal policy as stated in Theorem \ref{thm:threshold} and Theorem \ref{thm:threshold_t}.} 
\label{fig:threshold}
\end{figure}

\begin{figure}[t] 
\centering
\subfigure[$l \in \mathcal{L}^{(0)}$ for $\mu(l,1) = 2.1$ Mbps.]{
   \includegraphics[width=7cm]{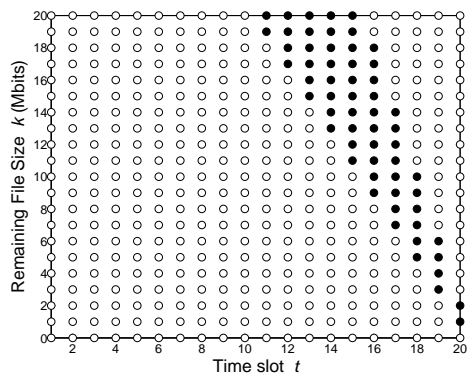}
   \label{fig:nonthreshold_fig1}
 }
\subfigure[$l \in \mathcal{L}^{(1)}$ for $\mu(l,1) = 3.1$ Mbps and $\mu(l,2) = 2.1$ Mbps.]{
   \includegraphics[width=7cm]{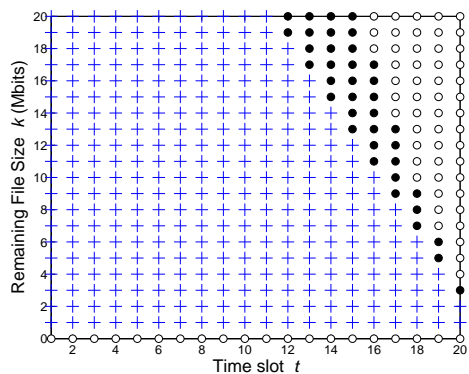}
   \label{fig:nonthreshold_fig2}
 } 
 \caption{An example of the optimal policy at location $l \in \mathcal{L}$ for the case with step penalty and location-dependent data rates, where $K = 20$ Mbits, $\sigma = 1$ Mbits, $T = 20$, and $Z = 100000$. The white dots ($\circ$), black dots ($\bullet$), and blue crosses (\textcolor{blue}{+}) represent the transmission decisions of $a = 0$ (idle), $1$ (use cellular), and $2$ (use Wi-Fi), respectively.} 
\label{fig:nonthreshold}
\end{figure}

\subsection{Demonstration of the Optimal Policy under Different Penalty Functions}

  In addition, we illustrate the actions of the optimal policy for different system states with a simple example. We first look at the special case with convex penalty function $h(k)$ and location-independent data rates $\mu_1$ and $\mu_2$, and cost $q = 1$ in \eqref{equ:cost} for $K = 20$ Mbits, $\sigma = 1$ Mbits, $T = 20$, and $b = 10$. 
  In Figures \ref{fig:threshold_fig1} and \ref{fig:threshold_fig2}, we can observe the threshold structure in dimensions $k$ and $t$ as stated in \eqref{equ:threshold0} and \eqref{equ:threshold_time0} for $l \in \mathcal{L}^{(0)}$ and in \eqref{equ:threshold2} and \eqref{equ:threshold_time2} for $l \in \mathcal{L}^{(1)}$ with $\mu_2 \leq \mu_1$ as stated in Theorem \ref{thm:threshold}. We can also notice the change in thresholds as stated in Theorem \ref{thm:threshold_t}. 
  
  Finally, we show an example of the optimal policy for the general case with non-convex penalty function $h(k)$ and location-dependent cellular/Wi-Fi data rates. 
  We consider a step penalty function $h(k) = Z$ for $k > 0$ and $h(0) = 0$, where $Z >> 1$ is a large positive constant. With this penalty function, the objective is to complete the file transfer with the minimal cellular usage. We adopt $Z = 100000$. 
  As shown in Fig. \ref{fig:nonthreshold}, we can see that multiple thresholds exist along dimension $k$, instead of a single threshold in the special case.
  For example, in Fig. \ref{fig:nonthreshold}(a), for $t \geq 16$, when $k$ is increased from zero, the decision first changes from idle to using cellular, because a complete file transfer is still possible.
  However, when $k$ is increased further that a complete file transfer is impossible, the idle action is chosen.
  Notice that it is very different from the policy in the special case as stated in Theorem \ref{thm:threshold}, where the MU would not stay idle even when there is no chance to complete the file transfer.
	To sum up, the penalty function has a significant impact on the optimal policy, and it should be chosen carefully according to the QoS requirement of the application.

\section{Conclusions} \label{sec:concl}

  In this paper, we studied the user-initiated Wi-Fi offloading problem for delay-tolerant applications under usage-based pricing. The user aims to minimize its total data usage payment, while taking into account the deadline of the file transfer.
  We first proposed a general DAWN algorithm for the general case using dynamic programming.
  We then established sufficient conditions under which the optimal policy has a threshold structure in both dimensions $k$ and $t$. As a result, we proposed a monotone DAWN algorithm with a lower complexity that approximately solves the general offloading problem.
  It should be noted that the proposed algorithms are highly non-trivial, and they cannot be obtained simply by a standard application of dynamic programming.
  Contrary to the practices in some heuristic schemes that favour offloading traffic to Wi-Fi networks whenever possible, 
  our simulation results showed that it is not always optimal for a user to perform Wi-Fi offloading when the deadline requirement is stringent and the data rate in the cellular network is much higher than that in the Wi-Fi network (e.g. a 4G LTE-A cellular system versus a congested Wi-Fi network).
  On the other hand, when the file transfer can be completed easily by the deadline, the delay-aware design in DAWN and Wiffler helps reduce the payment of the users.
  Overall, our results suggested that future cellular and Wi-Fi integration system should include dynamic offloading policies that take into account the users' QoS and the real-time network loads, instead of using simplistic and static offloading policies.

  As we considered the user-initiated offloading, where users are usually self-interested, we focused on the offloading decision of a single user. We believe that it is an important step towards a better understanding of the multi-user offloading problem. 
  In fact, we have made a step forward by considering the interactions of the network selection and data offloading decisions of multiple users in \cite{cheung_ca14}. However, in \cite{cheung_ca14}, we assumed that the mobility trajectory of each MU is estimated accurately, and we did not consider the delay-tolerant application with a given deadline. In other words, it is not possible to use the approach in \cite{cheung_ca14} to directly generalize the results  in this paper to the multi-user case.
  
	\rev{In this work, we have focused on the single file transfer by a given deadline. For future work, we will consider the case of multiple file transfers at the same time, and solve the problem by dynamic programming with additional states and decision variables. Considering the challenges of analyzing the single file case as in this paper, obtaining closed-form analysis and low complexity heuristic with clear engineering insights will be very challenging. 
	Moreover, in this paper, we consider Markovian user mobility model. It is an interesting direction for future research by considering other mobility models, especially the heavy-tailed distribution model \cite{rhee_ot11}, which has shown to be more accurate for modeling human mobilities.}

\bibliographystyle{IEEEtran}
\bibliography{IEEEabrv,mybibfile}

\begin{IEEEbiography}
[{\includegraphics[width=1in,height=1.25in,clip,keepaspectratio]{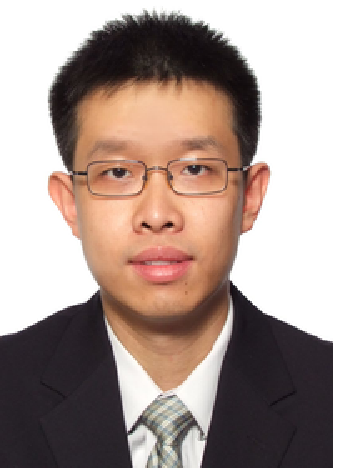}}]
{Man Hon Cheung} received the B.Eng. and M.Phil. degrees in Information Engineering from the Chinese University of Hong Kong (CUHK) in 2005 and 2007, respectively, and the Ph.D. degree in Electrical and Computer Engineering from the University of British Columbia (UBC) in 2012.
 Currently, he is a postdoctoral fellow in the Department of Information Engineering in CUHK.
 He received the IEEE Student Travel Grant for attending {\it IEEE ICC 2009}. He was awarded the Graduate Student International Research Mobility Award by UBC, and the Global Scholarship Programme for Research Excellence by CUHK.
 He serves as a Technical Program Committee member in {\it IEEE ICC}, {\it Globecom}, and {\it WCNC}.
 His research interests include the design and analysis of wireless network protocols using optimization theory, game theory, and dynamic programming, with current focus on mobile data offloading, mobile crowd sensing, and network economics.
\end{IEEEbiography}

\begin{IEEEbiography}[{\includegraphics[width=1in,height=1.25in,clip,keepaspectratio]{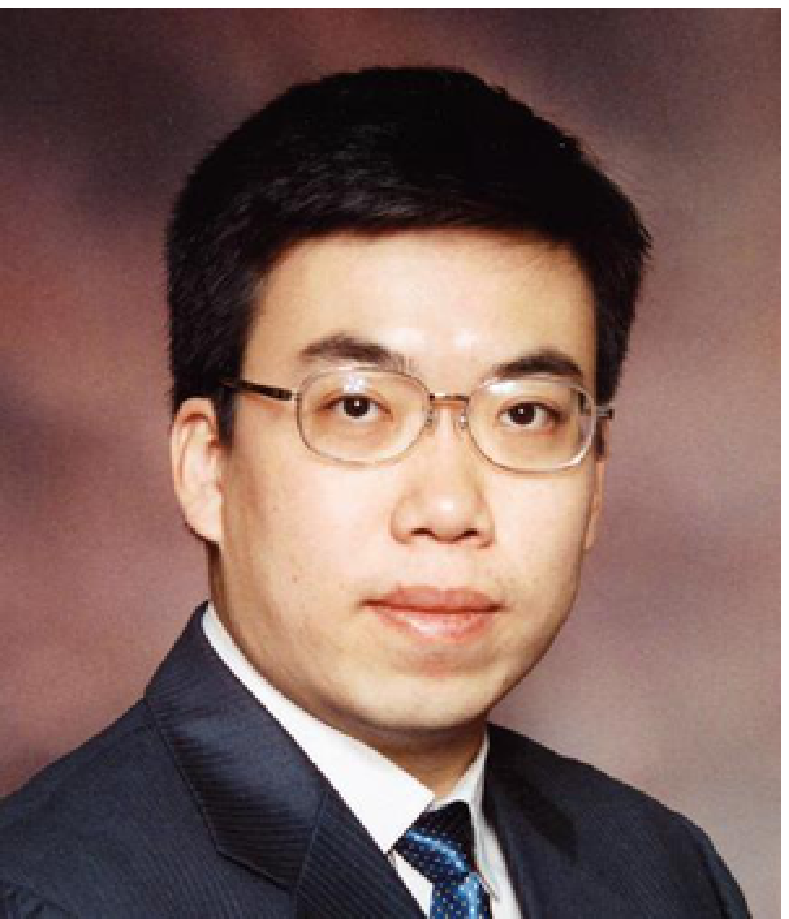}}]
{Jianwei Huang} (S'01-M'06-SM'11) is an Associate Professor and Director of the Network Communications and Economics Lab (ncel.ie.cuhk.edu.hk), in the Department of Information Engineering at the Chinese University of Hong Kong. He received the Ph.D. degree from Northwestern University in 2005. He is the co-recipient of 7 Best Paper Awards, including IEEE Marconi Prize Paper Award in Wireless Communications in 2011. He has co-authored four books: \emph{Wireless Network Pricing}, \emph{Monotonic Optimization in Communication and Networking Systems}, \emph{Cognitive Mobile Virtual Network Operator Games}, and \emph{Social Cognitive Radio Networks}. He has served as an Editor of IEEE Transactions on Cognitive Communications and Networking, IEEE Journal on Selected Areas in Communications - Cognitive Radio Series, and IEEE Transactions on Wireless Communications. He has served as a Guest Editor of IEEE Transactions on Smart Grid, IEEE Journal on Selected Areas in Communications, IEEE Communications Magazine, and IEEE Network. He has served as Associate Editor-in-Chief of IEEE Communications Society Technology News, Chair of IEEE Communications Society Multimedia Communications Technical Committee, and Vice Chair of IEEE Communications Society Cognitive Network Technical Committee. He is a Senior Member and a Distinguished Lecturer of IEEE Communications Society.
\end{IEEEbiography}


\newpage
\appendix  

%

\subsection{Proof of Lemma \ref{lem:vnondeck}} \label{app:vnondeck}   
 
\noindent (a) We prove it by induction. 
  First, from (\ref{equ:boundary}), $v_{T+1}(k,l) = h(k)$ is a non-decreasing function in $k$, $\forall \, l \in \mathcal{L}$.
  Assume that $v_{t+1}(k,l)$ is a non-decreasing function in $k$, $\forall \, l \in \mathcal{L}$.
  From (\ref{equ:opteq2_2}), since $p(l' \,|\, l) \geq 0, \, \forall \, l, l' \in \mathcal{L}$ and the function $I(a = 1)q$ is independent of $k$, $\psi_t(k,l,a)$ is a non-decreasing function in $k$, $\forall \, l \in \mathcal{L}, a \in \mathcal{A}$.   
  Thus, $v_t(k,l)$ in (\ref{equ:opteq}) is a non-decreasing function in $k$, $\forall \, l \in \mathcal{L}$. 
  
\noindent  (b) We prove it by induction. 
  First, for $t = T$, we have
\opt{opt1}{
\begin{equation}
	v_T(k,l) = \min_{a \in \mathcal{A}^{(l)}} \{ \psi_T(k, l, a) \} \leq \psi_T(k, l, 0) 
	= \sum_{l' \in \mathcal{L}} p(l' \,|\, l) v_{T+1}(k,l') 
	= h(k) = v_{T+1}(k,l).
\end{equation}
}
\opt{opt2}{
\begin{equation}
\begin{split}
	v_T(k,l) = \min_{a \in \mathcal{A}^{(l)}} \{ \psi_T(k, l, a) \} \leq \psi_T(k, l, 0) \\
	= \sum_{l' \in \mathcal{L}} p(l' \,|\, l) v_{T+1}(k,l')
	= h(k) = v_{T+1}(k,l).  \hspace{-0.5cm}
\end{split}	
\end{equation}
}
The first and second equalities are from \eqref{equ:opteq} and \eqref{equ:opteq2_2}, and the last two equalities are from \eqref{equ:boundary}.
  Assume that $v_{t+1}(k,l)$ is a non-decreasing function in $t$, $\forall \, k \in \mathcal{K}, l \in \mathcal{L}$.
  From \eqref{equ:opteq2_2}, since $p(l' \,|\, l) \geq 0, \, \forall \, l, l' \in \mathcal{L}$ and the function $I(a = 1)q$ is independent of $t$,  $\psi_t(k,l,a)$ is a non-decreasing function in $t$, $\forall \, k \in \mathcal{K}, l \in \mathcal{L}, a \in \mathcal{A}$.
  Thus, $v_t(k,l)$ in \eqref{equ:opteq} is a non-decreasing function in $t$, $\forall \, k \in \mathcal{K}, l \in \mathcal{L}$.
  \hfill \QEDclosed

\subsection{Proof of Lemma \ref{lem:wifi}} \label{app:wifi} 

	Let $k \in \mathcal{K}$ and $l \in \mathcal{L}$ be given.\\	
	(a) We have 
\opt{opt1}{
\begin{equation} \label{equ:comparepsi02}
	\psi_t(k,l,0) 
	= \sum_{l' \in \mathcal{L}} p(l' \,|\, l) \, v_{t+1}(k,l') 
	\geq \sum_{l' \in \mathcal{L}} p(l' \,|\, l) \, v_{t+1} \bigl( [k - \mu(l,2)]^+, l' \bigr)
	= \psi_t(k,l,2),
\end{equation}
}
\opt{opt2}{
\begin{equation} \label{equ:comparepsi02}
\begin{split}
	\psi_t(k,l,0) 
	= \sum_{l' \in \mathcal{L}} p(l' \,|\, l) \, v_{t+1}(k,l') \quad\quad\quad\quad\quad\quad\quad \\
	\geq \sum_{l' \in \mathcal{L}} p(l' \,|\, l) \, v_{t+1} \bigl( [k - \mu(l,2)]^+, l' \bigr)
	= \psi_t(k,l,2),
\end{split}
\end{equation}
}
where the two equalities are due to \eqref{equ:opteq2_2} and the inequality is due to Lemma \ref{lem:vnondeck}.\\   
  (b) First, since $\mu(l,1) \leq \mu(l,2)$, we have 
\opt{opt1}{
\begin{equation} \label{equ:comparepsi12}
	\psi_t(k,l,1)
	= q + \sum_{l' \in \mathcal{L}} p(l' \,|\, l) \, v_{t+1} \bigl( [k - \mu(l,1)]^+, l' \bigr) 
	\geq \sum_{l' \in \mathcal{L}} p(l' \,|\, l) \, v_{t+1} \bigl( [k - \mu(l,2)]^+, l' \bigr)
	= \psi_t(k,l,2),
\end{equation}
}
\opt{opt2}{
\begin{equation} \label{equ:comparepsi12}
\begin{split}
	\psi_t(k,l,1)
	= q + \sum_{l' \in \mathcal{L}} p(l' \,|\, l) \, v_{t+1} \bigl( [k - \mu(l,1)]^+, l' \bigr) \quad\\
	\geq \sum_{l' \in \mathcal{L}} p(l' \,|\, l) \, v_{t+1} \bigl( [k - \mu(l,2)]^+, l' \bigr)
	= \psi_t(k,l,2),
\end{split}	
\end{equation}
}
where the two equalities are due to \eqref{equ:opteq2_2} and the inequality is due to Lemma \ref{lem:vnondeck}.
  Combining the results from \eqref{equ:comparepsi02} and \eqref{equ:comparepsi12}, from \eqref{equ:deltat}, we have $\delta_t^*(k,l) = 2, \, \forall \, k \in \mathcal{K}, t \in \mathcal{T}$. \hfill \QEDclosed

\subsection{Superadditivity and Subadditivity of $\psi_t(k,l,a)$} \label{app:additive}   

  The proof of the threshold structure in dimension $k$ in Theorem \ref{thm:threshold} is based on the results in Lemmas \ref{lem:v_diff} and \ref{lem:subadditive}.
  Let $l \in \mathcal{L}$ be given. Let $\tilde{\mathcal{A}}^{(l)} = \{j,1\}$, where $j = 0$ if $l \in \mathcal{L}^{(0)}$ and $j = 2$ if $l \in \mathcal{L}^{(1)}$ as in \eqref{equ:setal_2}, and $\mu_0 = 0$. 
  With only two possible actions in $\tilde{\mathcal{A}}^{(l)}$, we can rewrite \eqref{equ:opteq2_2} as
\opt{opt1}{
\begin{equation} \label{equ:psi}
	\psi_t(k,l,a) = I(a = 1)q +  \sum_{l' \in \mathcal{L}} p(l' \,|\, l) \, \Bigl[ I(a = 1) v_{t+1} \bigl( [k - \mu_1]^+, l' \bigr) + \bigl(1 - I(a = 1)\bigr)  v_{t+1} \bigl( [k - \mu_j]^+, l' \bigr)  \Bigr].
\end{equation}
}
\opt{opt2}{
\begin{equation} \label{equ:psi}
\begin{split}
	\psi_t(k,l,a) = I(a = 1)q +  \sum_{l' \in \mathcal{L}} p(l' \,|\, l) \, \Bigl[ I(a = 1) \quad\quad\quad\quad\quad\quad \\ 
	\times v_{t+1} \bigl( [k - \mu_1]^+\!\!, l' \bigr) \! + \! \bigl(1 - I(a = 1)\bigr)  v_{t+1} \bigl( [k - \mu_j]^+\!\!, l' \bigr)  \Bigr]. \quad\quad
\end{split}	
\end{equation}
}

\begin{lemma} \label{lem:v_diff}
  If $\mu_j \leq \mu_1$ and $h(k)$ is a convex and non-decreasing function in $k$, then 
\opt{opt1}{
\begin{equation}
\begin{split}
  v_{t}( [ k - \mu_j ]^+ , l) - v_{t}\bigl( [ k - \mu_1 ]^+, l\bigr) 
  \geq v_{t}\bigl( [k - \sigma - \mu_j]^+ , l \bigr) - v_{t}\bigl([k - \sigma - \mu_1]^+,l\bigr), \\
  \quad\quad \forall \, k \in \mathcal{K}, l \in \mathcal{L}, t \in \mathcal{T} \cup \{T+1\}.
\end{split}
\end{equation}
}
\opt{opt2}{
\begin{equation}
\begin{split}
  v_{t}( [ k - \mu_j ]^+ , l) - v_{t}\bigl( [ k - \mu_1 ]^+, l\bigr)
  \geq v_{t}\bigl( [k - \sigma - \mu_j]^+ , l \bigr)  \\
   - v_{t}\bigl([k - \sigma - \mu_1]^+,l\bigr), 
  \forall \, k \in \mathcal{K}, l \in \mathcal{L}, t \in \mathcal{T} \cup \{T+1\}.
\end{split}
\end{equation}
}
\end{lemma}

\begin{proof}
  We prove it by induction. 
	Since $h(k)$ is a non-decreasing convex function, we have
\opt{opt1}{
\begin{equation} \label{equ:convex}
  h([ k - \mu_j ]^+) - h([k - \mu_1]^+) \geq h([k - \sigma - \mu_j]^+) - h([k - \sigma - \mu_1]^+), \quad \forall \, k \in \mathcal{K}.
\end{equation}
}
\opt{opt2}{
\begin{equation} \label{equ:convex}
\begin{split}
  h([ k - \mu_j ]^+) - h([k - \mu_1]^+) \geq h([k - \sigma - \mu_j]^+) \quad \\ - h([k - \sigma - \mu_1]^+), \quad \forall \, k \in \mathcal{K}.
\end{split}  
\end{equation}
}
  Let $k \in \mathcal{K}$, $l \in \mathcal{L}$ be given. For $t = T+1$, we have 
\opt{opt1}{
\begin{equation}
\begin{split}
  v_{T+1}([ k - \mu_j ]^+ , l) - v_{T+1}\bigl([ k - \mu_1 ]^+,l\bigr) 
  = h( [ k - \mu_j ]^+ ) - h\bigl([ k - \mu_1 ]^+ \bigr) \quad\quad\quad\quad\quad\quad\quad\quad \\
  \geq h([k - \sigma - \mu_j]^+) - h([k - \sigma - \mu_1]^+)  
  = v_{T+1}( [k - \sigma - \mu_j]^+ , l) - v_{T+1}\bigl([k - \sigma - \mu_1]^+,l\bigr),\\
\end{split}
\end{equation}
}
\opt{opt2}{
\begin{align} 
  	 & v_{T+1}([ k - \mu_j ]^+ , l) - v_{T+1}\bigl([ k - \mu_1 ]^+,l\bigr) \nonumber \\
= & \; h( [ k - \mu_j ]^+ ) - h\bigl([ k - \mu_1 ]^+ \bigr) \nonumber \\
\geq & \; h([k - \sigma - \mu_j]^+) - h([k - \sigma - \mu_1]^+) \nonumber \\
= & \; v_{T+1}( [k - \sigma - \mu_j]^+ , l) - v_{T+1}\bigl([k - \sigma - \mu_1]^+,l\bigr),
\end{align}
}
where the equalities are due to (\ref{equ:boundary}) and the inequality is due to (\ref{equ:convex}).
  Assume that for a given $t \in \mathcal{T}$, we have 
\opt{opt1}{
\begin{equation} \label{equ:v_diff_t+1}
\begin{split}
  v_{t+1}( [ k - \mu_j ]^+ , l) - v_{t+1}\bigl( [ k - \mu_1 ]^+, l\bigr) 
  \geq v_{t+1}\bigl( [k - \sigma - \mu_j]^+ , l \bigr) - v_{t+1}\bigl([k - \sigma - \mu_1]^+,l\bigr), \\
  \forall \, k \in \mathcal{K}, l \in \mathcal{L}.
\end{split}
\end{equation}  
}
\opt{opt2}{
\begin{equation} \label{equ:v_diff_t+1}
\begin{split}
  v_{t+1}\!( [ k \! - \! \mu_j ]^+ \!\!\!, l) \! - \! v_{t+1}\!\bigl( [ k \! - \! \mu_1 ]^+\!\!\!, l\bigr)
  \! \geq \! v_{t+1}\!\bigl( [k - \sigma - \mu_j]^+ \!\!\!, l \bigr) \\
  - v_{t+1}\bigl([k - \sigma - \mu_1]^+,l\bigr), 
  \forall \, k \in \mathcal{K}, l \in \mathcal{L}. \quad\quad\quad\quad\quad
\end{split}
\end{equation}
}

   From (\ref{equ:opteq}), let actions $a_1$, $a_2$, $a_3$, $a_4 \in \tilde{\mathcal{A}}^{(l)}$ be defined such that
\opt{opt1}{
\begin{eqnarray}
  v_{t}( [ k - \mu_j ]^+ , l) \!\!\! & = \!\!\! & \min_{a \in \tilde{\mathcal{A}}^{(l)}}\{ \psi_t([ k - \mu_j ]^+, l, a) \} = \psi_t([ k - \mu_j ]^+, l, a_1),\\ 			
  v_{t}\bigl( [ k - \mu_1 ]^+, l\bigr) \!\!\! & = \!\!\! & \min_{a \in \tilde{\mathcal{A}}^{(l)}}\{ \psi_t\bigl( [k - \mu_1 ]^+, l, a \bigr) \} = \psi_t\bigl( [ k - \mu_1 ]^+, l, a_2 \bigr), \label{equ:v_spt2} \\
  v_{t}\bigl( [ k - \sigma - \mu_j]^+, l\bigr) \!\!\! & = \!\!\! & \min_{a \in \tilde{\mathcal{A}}^{(l)}}\{ \psi_t\bigl( [ k - \sigma - \mu_j]^+, l, a \bigr) \} = \psi_t\bigl( [ k - \sigma - \mu_j]^+, l, a_3 \bigr), \text{ and} \label{equ:v_spt3} \\ 
  v_{t}\bigl( [ k - \sigma - \mu_1]^+, l\bigr) \!\!\! & = \!\!\! & \min_{a \in \tilde{\mathcal{A}}^{(l)}}\{ \psi_t\bigl( [ k - \sigma - \mu_1]^+, l, a \bigr) \} = \psi_t\bigl( [ k - \sigma - \mu_1]^+, l, a_4 \bigr). 
\end{eqnarray}
}
\opt{opt2}{
\begin{align}
	v_{t}( [ k - \mu_j ]^+ , l) & = \min_{a \in \tilde{\mathcal{A}}^{(l)}}\{ \psi_t([ k - \mu_j ]^+, l, a) \} \nonumber \\
	& = \psi_t([ k - \mu_j ]^+, l, a_1),
\end{align}
\begin{align} \label{equ:v_spt2}
	v_{t}\bigl( [ k - \mu_1 ]^+, l\bigr) & =  \min_{a \in \tilde{\mathcal{A}}^{(l)}}\{ \psi_t\bigl( [k - \mu_1 ]^+, l, a \bigr) \} \nonumber \\
	& = \psi_t\bigl( [ k - \mu_1 ]^+, l, a_2 \bigr), 
\end{align}
\begin{align} \label{equ:v_spt3}
	v_{t}\bigl( [ k - \sigma - \mu_j]^+, l\bigr) & = \min_{a \in \tilde{\mathcal{A}}^{(l)}}\{ \psi_t\bigl( [ k - \sigma - \mu_j]^+, l, a \bigr) \} \nonumber \\
	& = \psi_t\bigl( [ k - \sigma - \mu_j]^+, l, a_3 \bigr), \text{   and}  
\end{align}
\begin{align}
	v_{t}\bigl( [ k - \sigma - \mu_1]^+, l\bigr) & = \min_{a \in \tilde{\mathcal{A}}^{(l)}}\{ \psi_t\bigl( [ k - \sigma - \mu_1]^+, l, a \bigr) \} \nonumber \\
	& = \psi_t\bigl( [ k - \sigma - \mu_1]^+, l, a_4 \bigr).		
\end{align}
}

  We thus have  
\opt{opt1}{
\begin{align} \label{equ:abcd}
	   & v_{t}( [ k - \mu_j ]^+ , l) - v_{t}\bigl( [ k - \mu_1 ]^+, l\bigr)  
  - v_{t}\bigl( [k - \sigma - \mu_j]^+ , l \bigr) + v_{t}\bigl([k - \sigma - \mu_1]^+,l\bigr) \nonumber \\
= 	 & \; \psi_t( [ k - \mu_j ]^+ , l, a_1) - \psi_t\bigl( [ k - \mu_1 ]^+, l, a_2\bigr)  
  - \psi_t\bigl( [k - \sigma - \mu_j]^+ , l, a_3 \bigr) + \psi_t\bigl([k - \sigma - \mu_1]^+,l, a_4\bigr) \nonumber \\
= 	 & \; \underbrace{\psi_t( [ k - \mu_j ]^+ , l, a_1) - \psi_t\bigl( [k - \sigma - \mu_j]^+ , l, a_1 \bigr)}_{A}  \nonumber \\  
		 & +  \underbrace{\psi_t\bigl( [k - \sigma - \mu_j]^+ , l, a_1 \bigr) - \psi_t\bigl( [k - \sigma - \mu_j]^+ , l, a_3 \bigr)}_{B} \nonumber \\
		 & \underbrace{\Bigl( - \, \psi_t\bigl( [ k - \mu_1 ]^+, l, a_2\bigr) + \psi_t\bigl( [ k - \mu_1 ]^+, l, a_4\bigr) \Bigr)}_{C} \nonumber \\
		 & - \Bigl( \underbrace{\psi_t\bigl( [ k - \mu_1 ]^+, l, a_4\bigr) - \psi_t\bigl([k - \sigma - \mu_1]^+,l, a_4\bigr)}_{D} \Bigr) =  A + B + C - D. 
\end{align}
}
\opt{opt2}{
\begin{align} \label{equ:abcd}
	   & v_{t}( [ k - \mu_j ]^+ , l) - v_{t}\bigl( [ k - \mu_1 ]^+, l\bigr)   \nonumber \\
  	 & - v_{t}\bigl( [k - \sigma - \mu_j]^+ , l \bigr) + v_{t}\bigl([k - \sigma - \mu_1]^+,l\bigr) \nonumber \\
= 	 & \; \psi_t( [ k - \mu_j ]^+ , l, a_1) - \psi_t\bigl( [ k - \mu_1 ]^+, l, a_2\bigr) \nonumber \\ 
  	 & - \psi_t\bigl( [k - \sigma - \mu_j]^+ , l, a_3 \bigr) + \psi_t\bigl([k - \sigma - \mu_1]^+,l, a_4\bigr) \nonumber \\
= 	 & \; \underbrace{\psi_t( [ k - \mu_j ]^+ , l, a_1) - \psi_t\bigl( [k - \sigma - \mu_j]^+ , l, a_1 \bigr)}_{A} \nonumber \\
		 & + \underbrace{\psi_t\bigl( [k - \sigma - \mu_j]^+ , l, a_1 \bigr) - \psi_t\bigl( [k - \sigma - \mu_j]^+ , l, a_3 \bigr)}_{B} \nonumber \\  
		 & \underbrace{\Bigl(- \, \psi_t\bigl( [ k - \mu_1 ]^+, l, a_2\bigr) + \psi_t\bigl( [ k - \mu_1 ]^+, l, a_4\bigr) \Bigr)}_{C} \nonumber \\
		 & - \Bigl( \underbrace{\psi_t\bigl( [ k - \mu_1 ]^+, l, a_4\bigr) - \psi_t\bigl([k - \sigma - \mu_1]^+,l, a_4\bigr)}_{D} \Bigr) \nonumber \\
=    & \; A + B + C - D.
\end{align}
}

  We have
\opt{opt1}{
\begin{align} 
A =  & \sum_{l' \in \mathcal{L}} p(l' \,|\, l) \Bigl[ I(a_1 = 1) \bigl[ v_{t+1}([k - \mu_j - \mu_1]^+,l') - v_{t+1}([k - \sigma - \mu_j - \mu_1]^+,l') \bigr] \nonumber \\
     & \quad\quad\quad\quad\quad + \bigl( 1 - I(a_1 = 1) \bigr) \bigl[ v_{t+1}([ k - 2\mu_j ]^+,l') - v_{t+1}([k - \sigma - 2\mu_j]^+,l') \bigr] \Bigr] \nonumber \\ 
\geq &  \sum_{l' \in \mathcal{L}} p(l' \,|\, l) \bigl[ v_{t+1}([k - \mu_j - \mu_1]^+,l') - v_{t+1}([k - \sigma - \mu_j - \mu_1]^+,l') \bigr] \nonumber \\				
\geq & \sum_{l' \in \mathcal{L}} p(l' \,|\, l) \Bigl[ I(a_4 = 1) \bigl[ v_{t+1}([k - 2\mu_1]^+,l') - v_{t+1}([k - \sigma - 2\mu_1]^+,l') \bigr]  \nonumber \\
		 & + \bigl( 1 - I(a_4 = 1) \bigr) \bigl[ v_{t+1}([k - \mu_j - \mu_1]^+ ,l') - v_{t+1}([k - \sigma - \mu_j - \mu_1]^+,l') \bigr] \Bigr] = D,   
\end{align}
}
\opt{opt2}{
\begin{align} 
A =  & \sum_{l' \in \mathcal{L}} p(l' \,|\, l) \Bigl[ I(a_1 = 1) \bigl[ v_{t+1}([k - \mu_j - \mu_1]^+,l')  \nonumber \\
     & - v_{t+1}([k - \sigma - \mu_j - \mu_1]^+,l') \bigr] + \bigl( 1 - I(a_1 = 1) \bigr) \nonumber \\
     &  \times \bigl[ v_{t+1}([ k - 2\mu_j ]^+,l') - v_{t+1}([k - \sigma - 2\mu_j]^+,l') \bigr] \Bigr] \nonumber \\
\geq & \sum_{l' \in \mathcal{L}} p(l' \,|\, l) \bigl[ v_{t+1}([k - \mu_j - \mu_1]^+,l')  \nonumber \\
		 & \hspace{2cm} - v_{t+1}([k - \sigma - \mu_j - \mu_1]^+,l') \bigr]  \nonumber \\		
\geq & \sum_{l' \in \mathcal{L}} p(l' \,|\, l) \Bigl[ I(a_4 = 1) \bigl[ v_{t+1}([k - 2\mu_1]^+,l')  \nonumber \\
		 &  - v_{t+1}([k - \sigma - 2\mu_1]^+,l') \bigr] + \bigl( 1 - I(a_4 = 1) \bigr) \nonumber \\
		 & \hspace{-0.7cm} \times \!\! \bigl[ v_{t+1}([k - \mu_j - \mu_1]^+\!\!,l') - v_{t+1}([k \! - \! \sigma \!- \! \mu_j \! - \! \mu_1]^+\!\!,l') \bigr] \Bigr] \nonumber \\
	=  & \; D,   
\end{align}
}
where the two equalities are obtained by using \eqref{equ:psi} and the two inequalities are obtained due to the induction hypothesis in (\ref{equ:v_diff_t+1}).
  From (\ref{equ:v_spt3}) and (\ref{equ:v_spt2}), we have $B \geq 0$ and $C \geq 0$, respectively. Overall, from (\ref{equ:abcd}), we obtain
\opt{opt1}{
\begin{equation}
	   v_{t}( [ k - \mu_j ]^+ , l) - v_{t}\bigl( [ k - \mu_1 ]^+, l\bigr)  
  - v_{t}\bigl( [k - \sigma - \mu_j]^+ , l \bigr) + v_{t}\bigl([k - \sigma - \mu_1]^+,l\bigr) \geq 0,
\end{equation}
}
\opt{opt2}{
\begin{equation}
\begin{split}
	   v_{t}( [ k - \mu_j ]^+ , l) - v_{t}\bigl( [ k - \mu_1 ]^+, l\bigr)  
  - v_{t}\bigl( [k - \sigma - \mu_j]^+ , l \bigr) \\ + v_{t}\bigl([k - \sigma - \mu_1]^+,l\bigr) \geq 0, \quad\quad\quad
\end{split}
\end{equation}
}
which completes the proof. 
\end{proof}



\vspace{0.2cm}
\begin{lemma} \label{lem:subadditive}
  If $\mu_j \leq \mu_1$ and $\forall \, \hat{k}, \check{k} \in \mathcal{K}, \, l \in \mathcal{L}, \, t \in \mathcal{T}$ with $\hat{k} \geq \check{k}$, where
\opt{opt1}{
\begin{equation}
	v_{t+1}( [\hat{k} - \mu_j]^+ , l) - v_{t+1}\bigl([ \hat{k} - \mu_1 ]^+,l\bigr) \geq v_{t+1}( [\check{k} - \mu_j]^+ , l) - v_{t+1}\bigl([ \check{k} - \mu_1 ]^+,l\bigr),
\end{equation}
}
\opt{opt2}{
\begin{equation}
\begin{split}
	v_{t+1}( [\hat{k} - \mu_j]^+ , l) - v_{t+1}\bigl([ \hat{k} - \mu_1 ]^+,l\bigr) \quad\quad\quad \\ 
	\geq v_{t+1}( [\check{k} - \mu_j]^+ , l) - v_{t+1}\bigl([ \check{k} - \mu_1 ]^+,l\bigr),
\end{split}
\end{equation}
}
then $\psi_t(k,l,a)$ is subadditive on $\mathcal{K} \times \tilde{\mathcal{A}}^{(l)}$ for $j = 0$, and superadditive  on $\mathcal{K} \times \tilde{\mathcal{A}}^{(l)}$ for $j = 2$, $\forall \, t \in \mathcal{T}$, respectively.
\end{lemma}   

\begin{proof}
  Let $\hat{k}, \check{k} \in \mathcal{K}$, $\hat{a}, \check{a} \in \tilde{\mathcal{A}}^{(l)}$, $l \in \mathcal{L}$, and $t \in \mathcal{T}$ be given, where $\hat{k} \geq \check{k}$ and $\hat{a} \geq \check{a}$. Then
\opt{opt1}{
\begin{align} 
   & \psi_t( \hat{k} ,l, \hat{a} )  + \psi_t( \check{k} ,l, \check{a} ) - \psi_t( \hat{k} ,l, \check{a} ) - \psi_t( \check{k} ,l, \hat{a} ) \nonumber \\
=  & \sum_{l' \in \mathcal{L}} p(l' \,|\, l) \Bigl( I(\check{a} = 1) - I(\hat{a} = 1) \Bigr) \Bigl[ v_{t+1}( [\hat{k} - \mu_j]^+ , l) - v_{t+1}\bigl([ \hat{k} - \mu_1 ]^+,l\bigr) \nonumber \\
   & \quad\quad\quad\quad\quad\quad\quad\quad\quad\quad\quad\quad\quad - v_{t+1}( [\check{k} - \mu_j]^+ , l) + v_{t+1}\bigl([ \check{k} - \mu_1 ]^+,l\bigr) \Bigr],
\end{align} 
} 
\opt{opt2}{
\begin{align} 
   & \hspace{-0cm} \! \psi_t( \hat{k} ,l, \hat{a} )  + \psi_t( \check{k} ,l, \check{a} ) - \psi_t( \hat{k} ,l, \check{a} ) - \psi_t( \check{k} ,l, \hat{a} ) \nonumber \\
=  & \hspace{-0cm} \sum_{l' \in \mathcal{L}} p(l' \,|\, l) \Bigl( I(\check{a} = 1) - I(\hat{a} = 1) \Bigr) \Bigl[ v_{t+1}( [\hat{k} - \mu_j]^+ , l) - \nonumber \\
   & \hspace{-0.3cm} v_{t+1}\bigl([ \hat{k} \! - \! \mu_1 ]^+\!,l\bigr) \! - v_{t+1}( [\check{k} \! - \! \mu_j]^+\!, l) + v_{t+1}\bigl([ \check{k} \! - \! \mu_1 ]^+\!,l\bigr) \Bigr], \nonumber \\
\end{align}
}
where the equality is derived using (\ref{equ:psi}). Notice that $p(l' \,|\, l) \geq 0, \, \forall \, l, l' \in \mathcal{L}$.
  First, for $j = 0$, we have $\hat{a}, \check{a} \in \{0,1\}$, so $I(\check{a} = 1) \leq I(\hat{a} = 1)$. From the given condition in Lemma \ref{lem:subadditive} and Definition \ref{def:additive}, we conclude that $\psi_t(k,l,a)$ is subadditive on $\mathcal{K} \times \tilde{\mathcal{A}}^{(l)}$.  
  On the other hand, for $j = 2$, we have $\hat{a}, \check{a} \in \{1,2\}$, so $I(\check{a} = 1) \geq I(\hat{a} = 1)$. We can then conclude that $\psi_t(k,l,a)$ is superadditive on $\mathcal{K} \times \tilde{\mathcal{A}}^{(l)}$. 
\end{proof}


\subsection{Proof of Threshold Policy in Dimension $k$ in Theorem \ref{thm:threshold}} \label{app:threshold}

  
	We consider the case $0 \leq \mu_j \leq \mu_1$. Let $\hat{k}, \check{k} \in \mathcal{K}$, $l \in \mathcal{L}$, and $t \in \mathcal{T}$ be given. Let $\check{k} = [\hat{k} - z \sigma]^+$, where $z > 0$. If the condition of Theorem \ref{thm:threshold} is satisfied, by iteratively applying Lemma \ref{lem:v_diff}, we have
\opt{opt1}{
\begin{equation}
\begin{split} 
  v_{t}( [ \hat{k} - \mu_j ]^+ , l) - v_{t}\bigl( [ \hat{k} - \mu_1 ]^+, l\bigr) 
\geq v_{t}\bigl( [\hat{k} - \sigma - \mu_j]^+ , l \bigr) - v_{t}\bigl([\hat{k} - \sigma - \mu_1]^+,l\bigr) \geq \cdots  \\
  \geq v_{t}\bigl( [\hat{k} - z \sigma - \mu_j]^+ , l \bigr) - v_{t}\bigl([\hat{k} - z \sigma - \mu_1]^+,l\bigr) =  v_{t}\bigl( [ \check{k} - \mu_j ]^+ , l \bigr) - v_{t}\bigl([\check{k} - \mu_1]^+,l\bigr).
\end{split}
\end{equation}
}
\opt{opt2}{
\begin{align} 
  	 & v_{t}( [ \hat{k} - \mu_j ]^+ , l) - v_{t}\bigl( [ \hat{k} - \mu_1 ]^+, l\bigr)  \nonumber \\
\geq & \; v_{t}\bigl( [\hat{k} - \sigma - \mu_j]^+ , l \bigr) - v_{t}\bigl([\hat{k} - \sigma - \mu_1]^+,l\bigr)
\geq \cdots \nonumber \\
\geq & \; v_{t}\bigl( [\hat{k} - z \sigma - \mu_j]^+ , l \bigr) - v_{t}\bigl([\hat{k} - z \sigma - \mu_1]^+,l\bigr) \nonumber \\
= & \; v_{t}\bigl( [ \check{k} - \mu_j ]^+ , l \bigr) - v_{t}\bigl([\check{k} - \mu_1]^+,l\bigr).
\end{align}
}

  For $l \in \mathcal{L}^{(0)}$, we consider $j = 0$ (see Appendix \ref{app:additive}). Since $0 = \mu_0 < \mu_1$, $\psi_t(k,l,a)$ is subadditive on $\mathcal{K} \times \tilde{\mathcal{A}}^{(l)}$ from Lemma \ref{lem:subadditive}. From \cite[pp.\,104,\,115]{puterman_md05}, $\delta_t^{*}(k,l)$ is a monotone non-decreasing function in $k$. From \eqref{equ:setal_2} and \eqref{equ:deltat}, since $\delta_t^{*}(k,l) \in \tilde{\mathcal{A}}^{(l)} = \{0,1\}$, $\delta_t^{*}(k,l)$ is in the form of (\ref{equ:threshold0}). 
  
  Then, we consider $l \in \mathcal{L}^{(1)}$ for $\mu_2 \leq \mu_1$. 
  Since $j = 2$ (see Appendix \ref{app:additive}), $\psi_t(k,l,a)$ is superadditive on $\mathcal{K} \times \tilde{\mathcal{A}}^{(l)}$ from Lemma \ref{lem:subadditive}. From \cite[pp.\,104,\,115]{puterman_md05}, $\delta_t^{*}(k,l)$ is a monotone non-increasing function in $k$. From \eqref{equ:setal_2} and \eqref{equ:deltat}, as $\delta_t^{*}(k,l) \in \tilde{\mathcal{A}}^{(l)} = \{1,2\}$, $\delta_t^{*}(k,l)$ is in the form of (\ref{equ:threshold2}).  \hfill \QEDclosed

\subsection{Incremental Changes of $v_t(k,l)$} \label{app:vtkl_increment}

  The proof of the threshold structure in dimension $t$ in Theorem \ref{thm:threshold} is based on the results in Lemmas \ref{lem:v_diff_boundary} and \ref{lem:v_diff_time}, which establish that the incremental changes of $v_t(k,l)$ with respect to $k$ is non-decreasing in time $t$.
  
\begin{lemma} \label{lem:v_diff_boundary}
  If $h(k)$ is a convex and non-decreasing function in $k$, then we have 
\opt{opt1}{
\begin{equation}
	v_{T+1}([k-\mu_j]^+ , l) - v_{T+1}\bigl( [ k - \mu_1 ]^+, l\bigr)   
  \geq  v_{T}([k-\mu_j]^+ , l) - v_{T}\bigl( [ k - \mu_1 ]^+, l\bigr), \, \forall \, k \in \mathcal{K}, l \in \mathcal{L}.
\end{equation}
}
\opt{opt2}{
\begin{equation}
\begin{split}
	v_{T+1}([k-\mu_j]^+\!\!, l) - v_{T+1}\bigl( [ k - \mu_1 ]^+\!\!, l\bigr) \geq  v_{T}([k-\mu_j]^+\!\!, l) \\
	- v_{T}\bigl( [ k - \mu_1 ]^+, l\bigr), \, \forall \, k \in \mathcal{K}, l \in \mathcal{L}. \quad\quad\quad
\end{split}  
\end{equation}
}   
\end{lemma}

\begin{proof}
  First, by \eqref{equ:boundary}, we have
\opt{opt1}{
\begin{equation}
	\text{LHS} = v_{T+1}([k-\mu_j]^+ , l) - v_{T+1}\bigl( [ k - \mu_1 ]^+, l\bigr)  = h([k-\mu_j]^+) - h([ k - \mu_1 ]^+).
\end{equation}
}
\opt{opt2}{
\begin{equation}
\begin{split}
	\text{LHS} = v_{T+1}([k-\mu_j]^+ , l) - v_{T+1}\bigl( [ k - \mu_1 ]^+, l\bigr) \\
	 = h([k-\mu_j]^+) - h([ k - \mu_1 ]^+). \quad\quad\quad\quad\;\,
\end{split}
\end{equation}
}

  Next, we obtain
\opt{opt1}{
\begin{align} 
\text{RHS} = & \, v_{T}([k - \mu_j]^+ , l) - v_{T}\bigl( [ k - \mu_1 ]^+, l\bigr) \nonumber \\ 
= & \min\{\psi_T([k - \mu_j]^+,l,j), \psi_T([k-\mu_j]^+,l,1)\} - \min\{\psi_T([ k - \mu_1 ]^+,l,j), \psi_T([ k - \mu_1 ]^+,l,1)\} \nonumber \\				
= & \min\Bigl\{\sum_{l' \in \mathcal{L}} p(l' \,|\, l) \, v_{T+1} \bigl( [k - 2\mu_j]^+, l' \bigr), 
		q + \sum_{l' \in \mathcal{L}} p(l' \,|\, l) \, v_{T+1} \bigl( [k - \mu_j - \mu_1]^+, l' \bigr) \Bigr\}  \nonumber \\
	& - \min\Bigl\{\sum_{l' \in \mathcal{L}} p(l' \,|\, l) \, v_{T+1} \bigl( [k - \mu_j - \mu_1]^+, l' \bigr), 
		q + \sum_{l' \in \mathcal{L}} p(l' \,|\, l) \, v_{T+1} \bigl( [k - 2 \mu_1]^+, l' \bigr) \Bigr\}	\nonumber \\
= & \min\Bigl\{\!h([k - 2\mu_j]^+), q+h([k - \mu_j - \mu_1]^+) \!\Bigr\} - \min\Bigl\{\!h([k - \mu_j - \mu_1]^+), q+h([k - 2\mu_1]^+)\!\Bigr\}. 
\end{align}
}
\opt{opt2}{
\begin{align} 
 	& \text{RHS} \nonumber \\ 
= & v_{T}([k - \mu_j]^+ , l) - v_{T}\bigl( [ k - \mu_1 ]^+, l\bigr) \nonumber \\ 
= & \min\{\psi_T([k - \mu_j]^+,l,j), \psi_T([k-\mu_j]^+,l,1)\}  \nonumber \\
  & - \min\{\psi_T([ k - \mu_1 ]^+,l,j), \psi_T([ k - \mu_1 ]^+,l,1)\} \nonumber \\
= & \min\Bigl\{\sum_{l' \in \mathcal{L}} p(l' \,|\, l) \, v_{T+1} \bigl( [k - 2\mu_j]^+, l' \bigr), \nonumber \\
	&	q + \sum_{l' \in \mathcal{L}} p(l' \,|\, l) \, v_{T+1} \bigl( [k - \mu_j - \mu_1]^+, l' \bigr) \Bigr\}  \nonumber \\
	& - \min\Bigl\{\sum_{l' \in \mathcal{L}} p(l' \,|\, l) \, v_{T+1} \bigl( [k - \mu_j - \mu_1]^+, l' \bigr), \nonumber \\
	&	q + \sum_{l' \in \mathcal{L}} p(l' \,|\, l) \, v_{T+1} \bigl( [k - 2 \mu_1]^+, l' \bigr) \Bigr\}	\nonumber \\
= & \min\Bigl\{h([k - 2\mu_j]^+), q+h([k - \mu_j - \mu_1]^+) \Bigr\}  \nonumber \\
  & - \min\Bigl\{h([k - \mu_j - \mu_1]^+), q+h([k - 2\mu_1]^+)\Bigr\}.
\end{align}
}
%
  The second, third, and fourth equalities are due to \eqref{equ:opteq}, \eqref{equ:opteq2_2}, and \eqref{equ:boundary}, respectively.
   We consider the following two cases:

Case I: $q+h([k - \mu_j - \mu_1]^+) > h([k - 2\mu_j]^+)$. In this case, we have 
\opt{opt1}{
\begin{equation}
	q > h([k - 2\mu_j]^+) - h([k - \mu_j - \mu_1]^+) \geq h([k - \mu_j - \mu_1]^+) - h([k - 2\mu_1]^+),
\end{equation}
}
\opt{opt2}{
\begin{equation}
\begin{split}
	q > h([k - 2\mu_j]^+) - h([k - \mu_j - \mu_1]^+) \\
	\geq h([k - \mu_j - \mu_1]^+) - h([k - 2\mu_1]^+), \!
\end{split}	
\end{equation}
}
where the second inequality is due to the fact that $h(k)$ is a convex and non-decreasing function in $k$, and $\mu_j \leq \mu_1$. Thus, we obtain $q + h([k - 2\mu_1]^+) \geq h([k - \mu_j - \mu_1]^+)$.
  As a result, we have
\opt{opt1}{
\begin{equation} \label{equ:lrsrhs1}
	\text{RHS} = h([k - 2\mu_j]^+) - h([k - \mu_j - \mu_1]^+) 
	\leq h([k - \mu_j]^+) - h([k - \mu_1]^+) 
	= \text{LHS},
\end{equation}
}
\opt{opt2}{
\begin{equation} \label{equ:lrsrhs1}
\begin{split}
	\text{RHS} = h([k - 2\mu_j]^+) - h([k - \mu_j - \mu_1]^+) \\
	\leq h([k - \mu_j]^+) - h([k - \mu_1]^+) 
	= \text{LHS}, \hspace{-0.3cm}
\end{split}	
\end{equation}
}
where the inequality is established for convex and non-decreasing $h(k)$ and $\mu_j \leq \mu_1$.

Case II: $q+h([k - \mu_j - \mu_1]^+) \leq h([k - 2\mu_j]^+)$
  In this case, we have
\opt{opt1}{
\begin{equation}
	\text{RHS} = q+h([k - \mu_j - \mu_1]^+) - \min\Bigl\{h([k - \mu_j - \mu_1]^+), q+h([k - 2\mu_1]^+)\Bigr\}.
\end{equation}
}
\opt{opt2}{
\begin{equation}
\begin{split}
	\text{RHS} = q+h([k - \mu_j - \mu_1]^+) \quad\quad\quad\quad\quad\quad\quad\quad\;\\
	- \min\Bigl\{h([k - \mu_j - \mu_1]^+), q+h([k - 2\mu_1]^+)\Bigr\}.
\end{split}	
\end{equation}
}
  We consider the following two subcases in Case II:
(a) First, if $h([k - \mu_j - \mu_1]^+) \leq q+h([k - 2\mu_1]^+)$, then we have
\opt{opt1}{
\begin{equation} \label{equ:lrsrhs2}
\begin{split}
	\text{RHS} = q+h([k - \mu_j - \mu_1]^+) - h([k - \mu_j - \mu_1]^+) 
	\leq h([k - 2\mu_j]^+) - h([k - \mu_j - \mu_1]^+) \\
	\leq h([k - \mu_j]^+) - h([k - \mu_1]^+) 
	= \text{LHS},
\end{split}	
\end{equation}
}
\opt{opt2}{
\begin{equation} \label{equ:lrsrhs2}
\begin{split}
	\text{RHS} = q+h([k - \mu_j - \mu_1]^+) - h([k - \mu_j - \mu_1]^+) \\
	\leq h([k - 2\mu_j]^+) - h([k - \mu_j - \mu_1]^+) \hspace{1.2cm} \\
	\leq h([k - \mu_j]^+) - h([k - \mu_1]^+) 
	= \text{LHS}, \hspace{0.8cm}
\end{split}	
\end{equation}
}
where the first inequality is due to the given condition in Case II, and the second inequality is due to the convex and non-decreasing $h(k)$.

(b) Second, if $h([k - \mu_j - \mu_1]^+) > q+h([k - 2\mu_1]^+)$, then we have
\opt{opt1}{
\begin{equation} \label{equ:lrsrhs3}
	\text{RHS} = h([k - \mu_j - \mu_1]^+) - h([k - 2\mu_1]^+) \leq h([k - \mu_j]^+) - h([k - \mu_1]^+) = \text{LHS}, 
\end{equation}
}
\opt{opt2}{
\begin{equation} \label{equ:lrsrhs3}
\begin{split}
	\text{RHS} = h([k - \mu_j - \mu_1]^+) - h([k - 2\mu_1]^+) \\
	\leq h([k - \mu_j]^+) - h([k - \mu_1]^+) = \text{LHS}, \hspace{-0.4cm}
\end{split}	
\end{equation}
}
where the inequality is due to the fact that $h(k)$ is a convex and non-decreasing function in $k$.
  Combining all the cases in \eqref{equ:lrsrhs1}, \eqref{equ:lrsrhs2}, and \eqref{equ:lrsrhs3}, we have $\text{LHS} \geq \text{RHS}$.
\end{proof}

\begin{lemma} \label{lem:v_diff_time}
  If $h(k)$ is a convex and non-decreasing function in $k$, then we have 
\opt{opt1}{
\begin{equation}
	v_{t+1}([k - \mu_j]^+ , l) - v_{t+1}\bigl( [ k - \mu_1 ]^+, l\bigr)   
  \geq  v_{t}([k - \mu_j]^+ , l) - v_{t}\bigl( [ k - \mu_1 ]^+, l\bigr), \, \forall \, k \in \mathcal{K}, l \in \mathcal{L}, t \in \mathcal{T}.
\end{equation}
}
\opt{opt2}{
\begin{equation}
\begin{split}
	v_{t+1}([k - \mu_j]^+ , l) - v_{t+1}\bigl( [ k - \mu_1 ]^+, l\bigr)   
  \geq  v_{t}([k - \mu_j]^+ , l) \\
  - v_{t}\bigl( [ k - \mu_1 ]^+, l\bigr), \, \forall \, k \in \mathcal{K}, l \in \mathcal{L}, t \in \mathcal{T}. \quad\quad\quad\quad
\end{split}  
\end{equation}
}
\end{lemma}

\begin{proof}
	We prove it by induction. First, from Lemma \ref{lem:v_diff_boundary}, we have established the result for $t = T$.
  Assume that for a given $t \in \mathcal{T}$, we have
\opt{opt1}{
\begin{equation} \label{equ:v_diff_t+1_time}
\begin{split}
  v_{t+2}([k - \mu_j]^+ , l) - v_{t+2}\bigl( [ k - \mu_1 ]^+, l\bigr) 
  \geq v_{t+1}\bigl([k - \mu_j]^+ , l \bigr) - v_{t+1}\bigl([k - \mu_1]^+,l\bigr), \forall \, k \in \mathcal{K}, l \in \mathcal{L}.
\end{split}
\end{equation}
}
\opt{opt2}{
\begin{equation} \label{equ:v_diff_t+1_time}
\begin{split}
  v_{t+2}([k - \mu_j]^+ , l) - v_{t+2}\bigl( [ k - \mu_1 ]^+, l\bigr) \hspace{2.7cm} \\
  \geq v_{t+1}\bigl([k - \mu_j]^+\!\!, l \bigr) - v_{t+1}\bigl([k - \mu_1]^+\!\!,l\bigr), \forall \, k \in \mathcal{K}, l \in \mathcal{L}.
\end{split}
\end{equation}
}
%
   Let actions $a_5$, $a_6$, $a_7$, $a_8 \in \tilde{\mathcal{A}}^{(l)}$ be defined such that
\opt{opt1}{   
\begin{eqnarray}
  v_{t+1}([k - \mu_j]^+, l) \!\!\! & = \!\!\! & \min_{a \in \tilde{\mathcal{A}}^{(l)}}\{ \psi_{t+1}([k - \mu_j]^+, l, a) \} = \psi_{t+1}([k - \mu_j]^+, l, a_5),\\ 			
  v_{t+1}\bigl( [ k - \mu_1 ]^+, l\bigr) \!\!\! & = \!\!\! & \min_{a \in \tilde{\mathcal{A}}^{(l)}}\{ \psi_{t+1}\bigl( [k - \mu_1]^+, l, a \bigr) \} = \psi_{t+1}\bigl( [k - \mu_1]^+, l, a_6 \bigr), \label{equ:v_spt2_time} \\
  v_{t}\bigl([k - \mu_j]^+, l\bigr) \!\!\! & = \!\!\! & \min_{a \in \tilde{\mathcal{A}}^{(l)}}\{ \psi_t\bigl([k - \mu_j]^+, l, a \bigr) \} = \psi_t\bigl([k - \mu_j]^+, l, a_7 \bigr), \text{ and} \label{equ:v_spt3_time} \\ 
  v_{t}\bigl( [ k - \mu_1]^+, l\bigr) \!\!\! & = \!\!\! & \min_{a \in \tilde{\mathcal{A}}^{(l)}}\{ \psi_t\bigl( [k - \mu_1]^+, l, a \bigr) \} = \psi_t\bigl( [k - \mu_1]^+, l, a_8 \bigr). 
\end{eqnarray}
}
\opt{opt2}{
\begin{align}
	v_{t+1}([k - \mu_j]^+, l) & = \min_{a \in \tilde{\mathcal{A}}^{(l)}}\{ \psi_{t+1}([k - \mu_j]^+, l, a) \} \nonumber \\
	& = \psi_{t+1}([k - \mu_j]^+, l, a_5),
\end{align}
\begin{align} \label{equ:v_spt2_time}
	v_{t+1}\bigl( [ k - \mu_1 ]^+, l\bigr) & =  \min_{a \in \tilde{\mathcal{A}}^{(l)}}\{ \psi_{t+1}\bigl( [k - \mu_1]^+, l, a \bigr) \} \nonumber \\
	& = \psi_{t+1}\bigl( [k - \mu_1]^+, l, a_6 \bigr), 
\end{align}
\begin{align} \label{equ:v_spt3_time}
	v_{t}\bigl([k - \mu_j]^+, l\bigr) & = \min_{a \in \tilde{\mathcal{A}}^{(l)}}\{ \psi_t\bigl([k - \mu_j]^+, l, a \bigr) \} \nonumber \\
	& = \psi_t\bigl([k - \mu_j]^+, l, a_7 \bigr), \text{   and}  
\end{align}
\begin{align}
	v_{t}\bigl( [ k - \mu_1]^+, l\bigr) & = \min_{a \in \tilde{\mathcal{A}}^{(l)}}\{ \psi_t\bigl( [k - \mu_1]^+, l, a \bigr) \} \nonumber \\
	& = \psi_t\bigl( [k - \mu_1]^+, l, a_8 \bigr).		
\end{align}
}

  We thus have  
\opt{opt1}{
\begin{align} \label{equ:efgh}
	   & v_{t+1}([k - \mu_j]^+ , l) - v_{t+1}\bigl( [ k - \mu_1 ]^+, l\bigr) - v_{t}([k - \mu_j]^+ , l) + v_{t}\bigl( [ k - \mu_1 ]^+, l\bigr) \nonumber \\
= 	 & \; \psi_{t+1}([k - \mu_j]^+, l, a_5) - \psi_{t+1}\bigl( [k - \mu_1]^+, l, a_6 \bigr)  
  - \psi_t\bigl([k - \mu_j]^+, l, a_7 \bigr) + \psi_t\bigl( [k - \mu_1]^+, l, a_8 \bigr) \nonumber \\
= 	 & \; \underbrace{\psi_{t+1}([k - \mu_j]^+, l, a_5) - \psi_{t}([k - \mu_j]^+, l, a_5)}_{E}    
		  +  \underbrace{\psi_{t}([k - \mu_j]^+, l, a_5) - \psi_t\bigl([k - \mu_j]^+, l, a_7 \bigr)}_{F} \nonumber \\
		 & \underbrace{\Bigl( - \psi_{t+1}\bigl( [k - \mu_1]^+, l, a_6 \bigr) + \psi_{t+1}\bigl( [k - \mu_1]^+, l, a_8\bigr) \Bigr)}_{G} \nonumber \\
		 & - \Bigl( \underbrace{\psi_{t+1}\bigl( [k - \mu_1]^+, l, a_8\bigr) - \psi_t\bigl( [k - \mu_1]^+, l, a_8 \bigr)}_{H} \Bigr) =  E + F + G - H. 
\end{align}
}
\opt{opt2}{
\begin{align} \label{equ:efgh}
	   & v_{t+1}([k - \mu_j]^+ , l) - v_{t+1}\bigl( [ k - \mu_1 ]^+, l\bigr)   \nonumber \\
  	 & - v_{t}([k - \mu_j]^+ , l) + v_{t}\bigl( [ k - \mu_1 ]^+, l\bigr) \nonumber \\
= 	 & \; \psi_{t+1}([k - \mu_j]^+, l, a_5) - \psi_{t+1}\bigl( [k - \mu_1]^+, l, a_6 \bigr) \nonumber \\ 
  	 & - \psi_t\bigl([k - \mu_j]^+, l, a_7 \bigr) + \psi_t\bigl( [k - \mu_1]^+, l, a_8 \bigr) \nonumber \\
= 	 & \; \underbrace{\psi_{t+1}([k - \mu_j]^+, l, a_5) - \psi_{t}([k - \mu_j]^+, l, a_5)}_{E} \nonumber \\
		 & +  \underbrace{\psi_{t}([k - \mu_j]^+, l, a_5) - \psi_t\bigl([k - \mu_j]^+, l, a_7 \bigr)}_{F} \nonumber \\  
		 & \underbrace{\Bigl( - \psi_{t+1}\bigl( [k - \mu_1]^+, l, a_6 \bigr) + \psi_{t+1}\bigl( [k - \mu_1]^+, l, a_8\bigr) \Bigr)}_{G} \nonumber \\
		 & - \Bigl( \underbrace{\psi_{t+1}\bigl( [k - \mu_1]^+, l, a_8\bigr) - \psi_t\bigl( [k - \mu_1]^+, l, a_8 \bigr)}_{H} \Bigr) \nonumber \\
=    & \; E + F + G - H.
\end{align}
}
  We have
\opt{opt1}{
\begin{align} 
E =  & \sum_{l' \in \mathcal{L}} p(l' \,|\, l) \Bigl[ I(a_5 = 1) \bigl[ v_{t+2}([k - \mu_j - \mu_1]^+,l') - v_{t+1}([k - \mu_j - \mu_1]^+,l') \bigr]  \nonumber \\
     & \quad\quad\quad\quad\quad + \bigl(1 - I(a_5 = 1)\bigr) \bigl[ v_{t+2}([k - 2\mu_j]^+,l') - v_{t+1}([k - 2\mu_j]^+,l') \bigr] \Bigr] \nonumber \\ 
\geq &  \sum_{l' \in \mathcal{L}} p(l' \,|\, l) \bigl[ v_{t+2}([k - \mu_j - \mu_1]^+,l') - v_{t+1}([k - \mu_j - \mu_1]^+,l') \bigr] \nonumber \\				
\geq & \sum_{l' \in \mathcal{L}} p(l' \,|\, l) \Bigl[ I(a_8 = 1) \bigl[ v_{t+2}([k - 2\mu_1]^+,l') - v_{t+1}([k - 2\mu_1]^+,l') \bigr] \nonumber \\
     & \quad\quad\quad\quad\quad + \bigl(1 - I(a_8 = 1)\bigr) \bigl[ v_{t+2}([k - \mu_j - \mu_1]^+,l') - v_{t+1}([k - \mu_j - \mu_1]^+,l') \bigr] \Bigr] = H,   
\end{align}
}
\opt{opt2}{
\begin{align} 
E =  & \sum_{l' \in \mathcal{L}} p(l' \,|\, l) \Bigl[ I(a_5 = 1) \bigl[ v_{t+2}([k - \mu_j - \mu_1]^+,l')  \nonumber \\
     & - v_{t+1}([k - \mu_j - \mu_1]^+,l') \bigr] + \bigl(1 - I(a_5 = 1)\bigr) \nonumber \\
     &  \times \bigl[ v_{t+2}([k - 2\mu_j]^+,l') - v_{t+1}([k - 2\mu_j]^+,l') \bigr] \Bigr] \nonumber \\
\geq & \sum_{l' \in \mathcal{L}} p(l' \,|\, l) \bigl[ v_{t+2}([k - \mu_j - \mu_1]^+,l')  \nonumber \\
		 & \hspace{2cm} - v_{t+1}([k - \mu_j - \mu_1]^+,l') \bigr]  \nonumber \\		
\geq & \sum_{l' \in \mathcal{L}} p(l' \,|\, l) \Bigl[ I(a_8 = 1) \bigl[ v_{t+2}([k - 2\mu_1]^+,l')  \nonumber \\
		 &  - v_{t+1}([k - 2\mu_1]^+,l') \bigr] + \bigl(1 - I(a_8 = 1)\bigr) \nonumber \\
		 & \hspace{-0.7cm} \times \!\! \bigl[ v_{t+2}([k - \mu_j - \mu_1]^+,l') - v_{t+1}([k - \mu_j - \mu_1]^+,l') \bigr] \Bigr] \nonumber \\
	=  & \; H,   
\end{align}
}
where the two equalities are obtained by using \eqref{equ:psi} and the two inequalities are due to the induction hypothesis in (\ref{equ:v_diff_t+1_time}).
  From (\ref{equ:v_spt3_time}) and (\ref{equ:v_spt2_time}), we have $F \geq 0$ and $G \geq 0$, respectively. Overall, from (\ref{equ:efgh}), we obtain
\opt{opt1}{
\begin{equation}
v_{t+1}([k - \mu_j]^+ , l) - v_{t+1}\bigl( [ k - \mu_1 ]^+, l\bigr)   
  \geq  v_{t}([k - \mu_j]^+ , l) - v_{t}\bigl( [ k - \mu_1 ]^+, l\bigr),
\end{equation}
}
\opt{opt2}{
\begin{equation}
\begin{split}
v_{t+1}([k - \mu_j]^+ , l) - v_{t+1}\bigl( [ k - \mu_1 ]^+, l\bigr)   \\
  \geq  v_{t}([k - \mu_j]^+ , l) - v_{t}\bigl( [ k - \mu_1 ]^+, l\bigr), \hspace{-0.5cm}
\end{split}
\end{equation}
}
which completes the proof.
\end{proof}

\subsection{Proof of Threshold Policy in Dimension $t$ in Theorem \ref{thm:threshold}} \label{app:threshold_time}

 Assume that there exists $t \in \mathcal{T}$ such that $\psi_t(k,l,1) < \psi_t(k,l,j)$. In this way, we have $\delta_t^{*}(k,l) = 1$ from \eqref{equ:deltat} and 
\opt{opt1}{
\begin{equation}
\begin{split}
	q < \sum_{l' \in \mathcal{L}} p(l' \,|\, l) \, \Bigl[ v_{t+1} \bigl( [k - \mu_j]^+, l' \bigr) - v_{t+1} \bigl( [k - \mu_1]^+, l' \bigr) \Bigr] \\
    \leq \sum_{l' \in \mathcal{L}} p(l' \,|\, l) \, \Bigl[ v_{t+2} \bigl( [k - \mu_j]^+, l' \bigr) - v_{t+2} \bigl( [k - \mu_1]^+, l' \bigr) \Bigr],
\end{split}
\end{equation}
}
\opt{opt2}{
\begin{equation}
\begin{split}
	q < \!\! \sum_{l' \in \mathcal{L}} p(l' \,|\, l) \Bigl[ v_{t+1} \bigl( [k - \mu_j]^+, l' \bigr) - v_{t+1} \bigl( [k - \mu_1]^+, l' \bigr) \Bigr] \\
    \leq \!\! \sum_{l' \in \mathcal{L}} p(l' \,|\, l) \Bigl[ v_{t+2} \bigl( [k - \mu_j]^+, l' \bigr) - v_{t+2} \bigl( [k - \mu_1]^+, l' \bigr) \Bigr],
\end{split}
\end{equation}
}
where the first inequality is by the definition in \eqref{equ:opteq2_2}, and the second inequality is from Lemma \ref{lem:v_diff_time}. It implies that $\psi_{t+1}(k,l,1) < \psi_{t+1}(k,l,j)$, so $\delta_{t+1}^{*}(k,l) = 1$ from \eqref{equ:deltat}. 
  Overall, we show that if there exists $t \in \mathcal{T}$ such that $\delta_t^{*}(k,l) = 1$, then $\delta_{t+1}^{*}(k,l) = 1$, which establishes the threshold structure of the optimal policy in the time dimension.
\hfill \QEDclosed

\subsection{Proof of Theorem \ref{thm:threshold_t}} \label{app:threshold_t}

  Let $j = 0$ for $l \in \mathcal{L}^{(0)}$ and $j = 2$ for $l \in \mathcal{L}^{(1)}$ as mentioned in Appendix \ref{app:additive}.
  
  (a) Let $l \in \mathcal{L}$ and $t \in \mathcal{T}$ be given.
  By the definition of threshold $k^*(l,t)$ in \eqref{equ:threshold0} and \eqref{equ:threshold2}, we have $\delta_t^*(k,l) = j$ for $0 \leq k < k^*(l,t)$.
  From the threshold structure in time in \eqref{equ:threshold_time0} and \eqref{equ:threshold_time2}, it implies that $\delta_{t-1}^*(k,l) = j$ for $0 \leq k < k^*(l,t)$.
  By the definition of threshold $k^*(l,t-1)$ at time $t-1$, we can conclude that $k^*(l,t-1) \geq k^*(l,t)$. 
  
  (b) Let $l \in \mathcal{L}$ and $k \in \mathcal{K}$ be given.
  By the definition of threshold $t^*(k,l)$ in \eqref{equ:threshold_time0} and \eqref{equ:threshold_time2}, we have $\delta_t^*(k,l) = 1$ for $t \geq t^*(k,l)$.
  It implies that $\delta_{t}^*(k + \sigma,l) = 1$ for $t \geq t^*(k,l)$ due to the threshold structure in file size in \eqref{equ:threshold0} and \eqref{equ:threshold2}.
  By the definition of threshold $t^*(k + \sigma,l)$ for file size $k + \sigma$, we can conclude that $t^*(k, l) \geq t^*(k + \sigma, l)$. \hfill \QEDclosed

\end{document}